\pdfoutput=1
\newif\ifFull
\Fulltrue

\ifFull
\documentclass[11pt,runningheads]{llncs}
\usepackage[margin=1in]{geometry}
\else
\documentclass{llncs}
\fi

\usepackage{hyperref}
\usepackage{graphicx}
\usepackage{subfigure}
\usepackage{enumerate}
\usepackage{times}
\usepackage{cite}
\usepackage{amsmath,amssymb}

\newcommand{\RR}{{\mathbb R}}
\newcommand{\Oh}{{\ensuremath{\mathcal{O}}}}
\newcommand{\Th}{{\Theta}}
\newcommand{\eps}{\ensuremath{\varepsilon}}

\renewenvironment{proof}{\medskip\noindent\textbf{Proof:}}{\mbox{}\hfill\qed\par\medskip}
\newenvironment{proofof}[1]{\medskip\noindent\textbf{Proof of #1:}}{\mbox{}\hfill\qed\par\medskip}

\title{Balanced Circle Packings for Planar Graphs}

\author{
Md.~Jawaherul~Alam\inst{1}
\and
David Eppstein\inst{2}
\and
Michael~T.~Goodrich\inst{2}
\and\\
Stephen~G.~Kobourov\inst{1}
\and
Sergey Pupyrev\inst{1,3}
}

\authorrunning{J.~Alam et al.}

\institute{
    Department of Computer Science, University of Arizona, Tucson, Arizona, USA
\and
    Department of Computer Science, University of California, Irvine, California, USA
\and
  Institute of Mathematics and Computer Science, Ural Federal University, Russia
}

\begin{document}

\maketitle

\begin{abstract}
We study \emph{balanced} circle packings and circle-contact representations
for planar graphs,
where
the ratio of the largest circle's diameter to the smallest circle's diameter
is polynomial in the number of circles.
We provide a number of positive
and negative results for the existence
of such balanced configurations.
%
\end{abstract}

\section{Introduction}

\ifFull
Circle packings are a frequently used and important tool in graph drawing~\cite{Malitz1994,BerEpp-WADS-01,Moh-GD-99,Keszegh2010,Rot-GD-11,AicRotSch-CGTA-12,Epp-GD-12,BekRaf-GD-12,EppHolLof-GD-13}.
\fi
In this application, they can be formalized using the notion of
a \emph{circle-contact representation} for a planar graph; this is a collection of interior-disjoint circles in $\RR^2$, corresponding one-for-one with the vertices of the graph, such that two vertices are adjacent if and only if their corresponding two circles are
\ifFull
tangent to each other~\cite{h-contact-96,Hli98}.
\fi
\ifFull
Graphs with circle-contact representations are also known as ``coin graphs''~\cite{Sachs1994133}.
\fi
In a classic paper, Koebe~\cite{Koebe36}
proved that every triangulated planar graph
has a circle-contact representation, and this has been subsequently
re-proved several times.
Generalizing this,
every planar graph has a circle-contact representation:
we can triangulate the graph by adding ``dummy'' vertices connected to the existing vertices within each face, produce a circle-contact representation for this augmented graph,
and then remove the circles corresponding to dummy vertices.
It is not always possible to describe a  circle-contact representation for a given graph
by a symbolic formula involving radicals~\cite{bs-rpg-93,Galois},
but they can nevertheless be constructed numerically and efficiently by
polynomial-time iterative schemes~\cite{ColSte-CGTA-03,Moh-DM-93}.

One of the drawbacks of some of these constructions, however,
is that the sizes of the circles in some of these configurations
may vary exponentially,
leading to drawings with very high area or with portions that are
so small that they are below the resolution of the display.
For this reason, we are interested in
\emph{balanced} circle packings and
circle-contact representations for planar graphs,
where the ratio of the
maximum and minimum diameters for the set of circles is
polynomial in the number of vertices in the graph; see \autoref{fig:motivation}.
\ifFull
Such drawings could be drawn with polynomial area, for instance,
where the smallest circle determines the minimum resolution.
\fi

\begin{figure}[t]
\centering
\includegraphics[width=\textwidth]{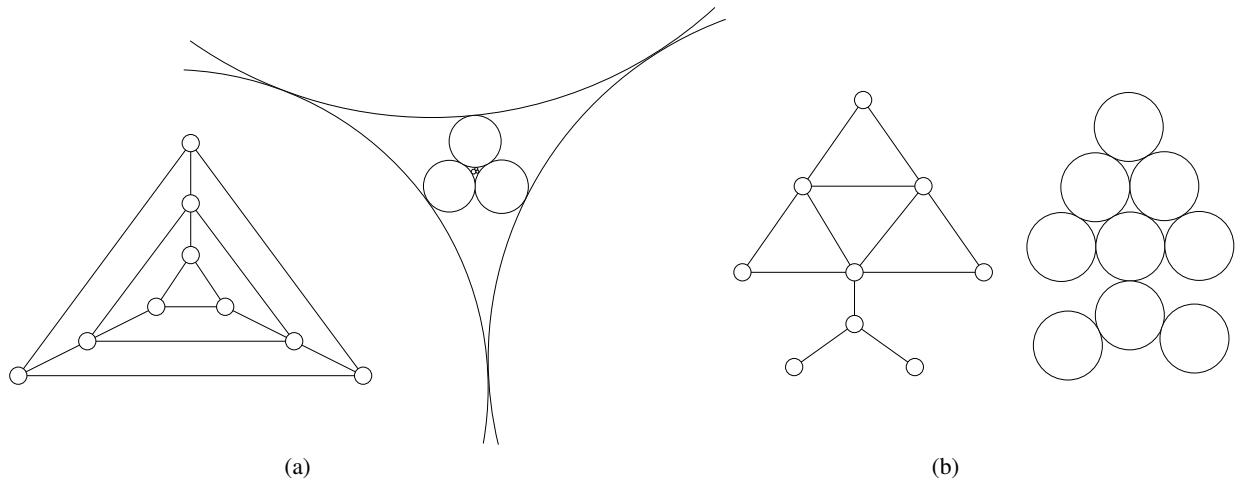}\\
(a)\hspace{0.5\textwidth}(b)
\caption{Two planar graphs with possible circle-contact representations: (a) a representation that is not optimally balanced; (b) a perfectly-balanced representation}
\label{fig:motivation}
\end{figure}

\ifFull
\subsection{Related Work}
\fi
There is a large body of work about representing
planar graphs as contact graphs,
where vertices are represented by geometrical objects
and edges correspond to two objects touching in some pre-specified fashion.
For example,
\ifFull
Hlin\v{e}n\'{y}~\cite{h-contact-96,Hli98} studies
\fi
contact representations using curves and line segments as
objects.
Several authors have considered contact graphs of triangles of various types.
For instance,
de~Fraysseix \textit{et al.}~\cite{FMR94} show that
every planar graph has a triangle-contact representation,
\ifFull
Badent \textit{et al.}~\cite{homothetic07} show that
partial planar 3-trees and some series-parallel graphs
have contact representations
with homothetic triangles,
\fi
and Gon\c{c}alves \textit{et al.}~\cite{GonLevPin-DCG-12}
prove that every 3-connected
planar graph and its dual can be simultaneously
represented by touching triangles (and
they point out that 4-connected planar graphs also
have contact representations with
homothetic triangles).
Also, Duncan \textit{et al.}~\cite{ghkk10} show that
every planar graph has a contact representation
with convex hexagons all of whose sides have one of three possible
slopes, and that hexagons are necessary for some graphs, if convexity is required.
With respect to balanced circle-contact representations,
Breu and Kirkpatrick~\cite{Breu19983} show that it is NP-complete to
test whether a graph has a perfectly-balanced
circle-contact representation, in which every circle is the same size.
\ifFull
Circle-contact graphs are related to \emph{disk graphs}~\cite{HK01}, which
represent a graph by intersecting disks; unlike for circle contact graphs, the interiors of the disks are not required to be
disjoint.
Regarding the resolution of disk graphs,
McDiarmid and M{\"{u}}ller~\cite{McDiarmid2013114} show that
there are $n$-vertex graphs such that in every realization by disks
with integer radii, at least one coordinate or radius is $2^{2^{\Omega(n)}}$,
and they also show that every disk graph can be realized by
disks with integer coordinates and radii that are at most $2^{2^{\Oh(n)}}$.
\fi

\ifFull
\subsection{New Results}
\fi
In this paper, we provide a number of positive and negative results
regarding
balanced circle-contact representations for planar graphs:
\begin{itemize}
\item Every planar graph with bounded maximum vertex degree and
 logarithmic outerplanarity admits a balanced circle-contact representation.
\item There exist planar graphs with bounded maximum degree and linear outerplanarity, or with linear maximum degree and bounded outerplanarity, that do not admit a balanced circle-contact representation.
\item Every tree admits a balanced circle-contact representation.
\item Every outerpath admits a balanced circle-contact representation.
\item Every cactus graph admits a balanced circle-contact representation.
\item Every planar graph with bounded tree-depth admits a balanced circle-contact representation.
\end{itemize}

\section{Bounded Degree and Logarithmic Outerplanarity}

A plane graph
(that is, a combinatorially fixed planar embedding of a planar graph)
is \emph{outerplanar} if
all of its vertices are on the outer face.
A \emph{$k$-outerplanar graph} is defined recursively.
As a base case, if a plane graph is outerplanar,
then it is a $1$-outerplanar graph.
A plane graph is $k$-outerplanar, for $k>1$,
if the removal of all the
outer vertices (and their incident edges)
yields a graph such that each of the remaining components
is $(k-1)$-outerplanar.
The \emph{outerplanarity} of a plane graph $G$ is the minimum
value for $k$ such that $G$ is $k$-outerplanar.

\ifFull
In this section, we show that every $n$-vertex plane graph with
bounded maximum degree
and $\Oh(\log n)$ outerplanarity admits a balanced
circle-contact representation.
We also show that these two restrictions are necessary, by
demonstrating examples of planar
 graphs where either of these restrictions are
violated and there is no balanced circle-contact representation.
\fi

\subsection{Balanced Circle-Contact Representations}

\begin{theorem}
\label{th:bounded}
Every $n$-vertex $k$-outerplanar graph with maximum degree $\Delta$
admits a circle-contact representation where the ratio of the maximum
and the minimum diameter is at most $f(\Delta)^{k+\log n}$,
for some positive function~$f$.
In particular, when $\Delta$ is a fixed constant
and $k$ is $\Oh(\log n)$, this ratio is polynomial in~$n$.
\end{theorem}

In order to prove the theorem,
we need the following result from~\cite{Malitz1994}.
\ifFull
See also~\cite{Keszegh2010}.
\fi

\begin{lemma}[Malitz-Papakostas]
\label{lm:bounded-packing}
The vertices of every triangulated planar graph $G$ with the maximum degree
$\Delta$ can be represented by nonoverlapping disks in the plane so that two disks are tangent to each other if and only if the corresponding vertices are adjacent, and
for each two disks that are tangent to each other, the ratio of the radii of the smaller to the larger disk
 is at least $\alpha^{\Delta-2}$ with $\alpha = \frac{1}{3+2\sqrt{3}} \approx 0.15$.
\end{lemma}

As a direct corollary,
every maximal planar graph with maximum degree
$\Delta=\Oh(1)$ and diameter $d=\Oh(\log n)$
has a balanced circle-contact representation.
\autoref{th:bounded} goes beyond this.

\begin{proofof}{\autoref{th:bounded}}
To prove the claim, it is sufficient to show how to augment a given
$k$-outerplanar graph
into a maximal planar graph with additional vertices so that its maximum degree remains
$\Oh(\Delta)$ and its diameter becomes $\Oh(k+\log n)$. By \autoref{lm:bounded-packing},
the resulting graph admits a circular contact representation with the given bounds on the
ratio of radii. Removing the circles corresponding to the added
vertices yields the desired balanced representation of the original graph.

\begin{figure}[b!]
\centering
\subfigure[]{
	\includegraphics[width=0.2\textwidth]{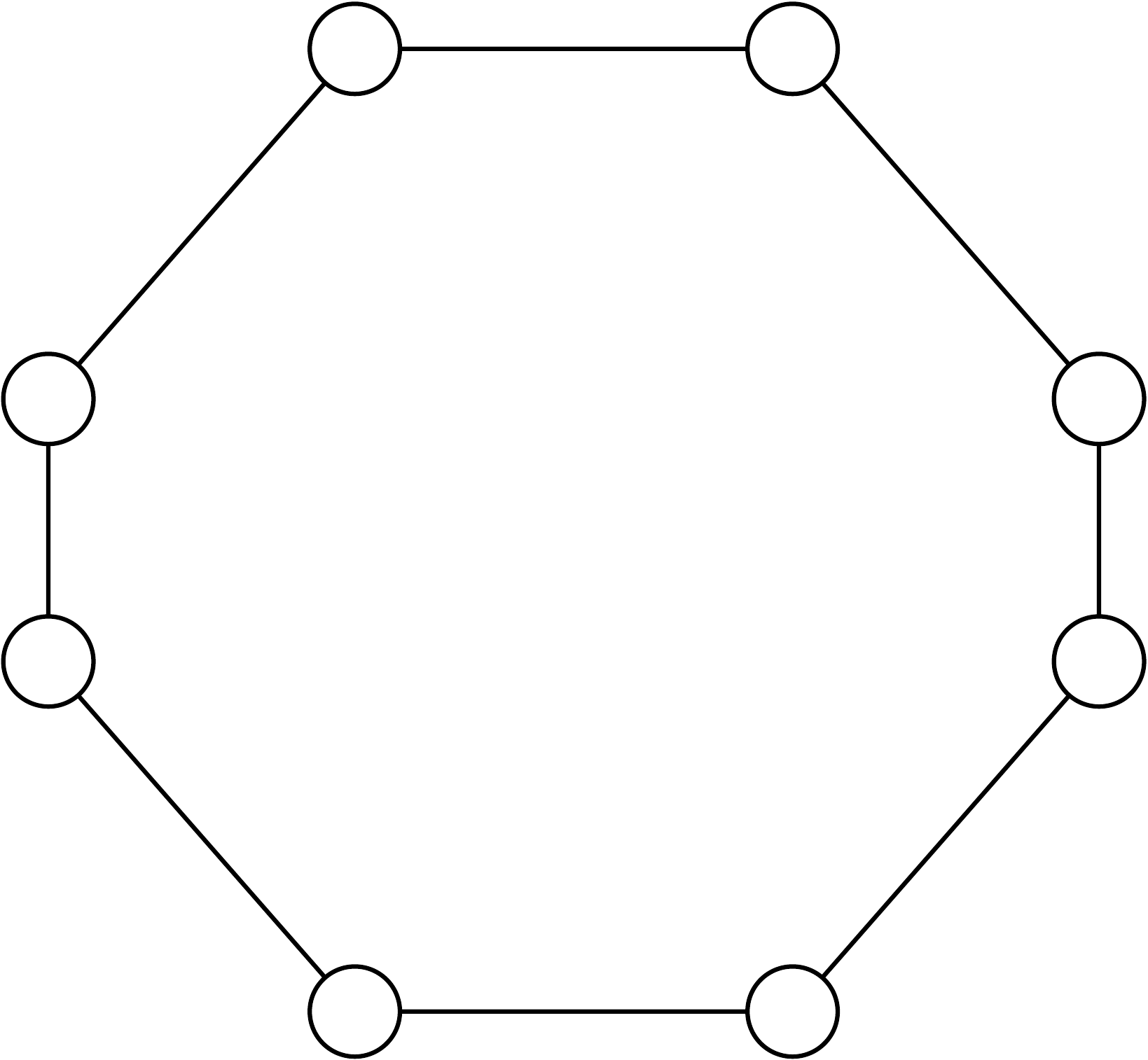}}
~~~~~~~~~~~~
\subfigure[]{
	\includegraphics[width=0.2\textwidth]{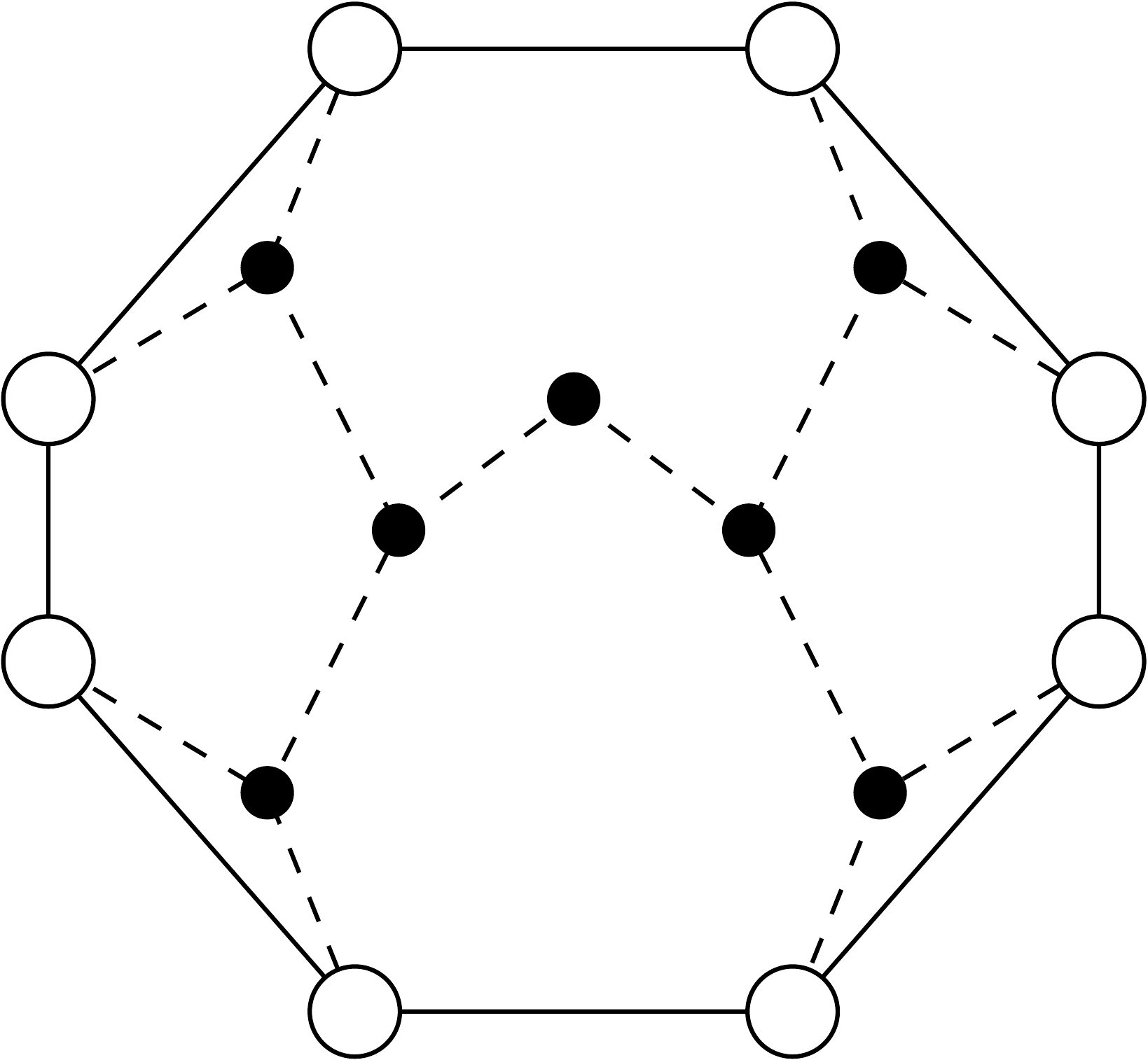}}
~~~~~~~~~~~~
\subfigure[]{
	\includegraphics[width=0.2\textwidth]{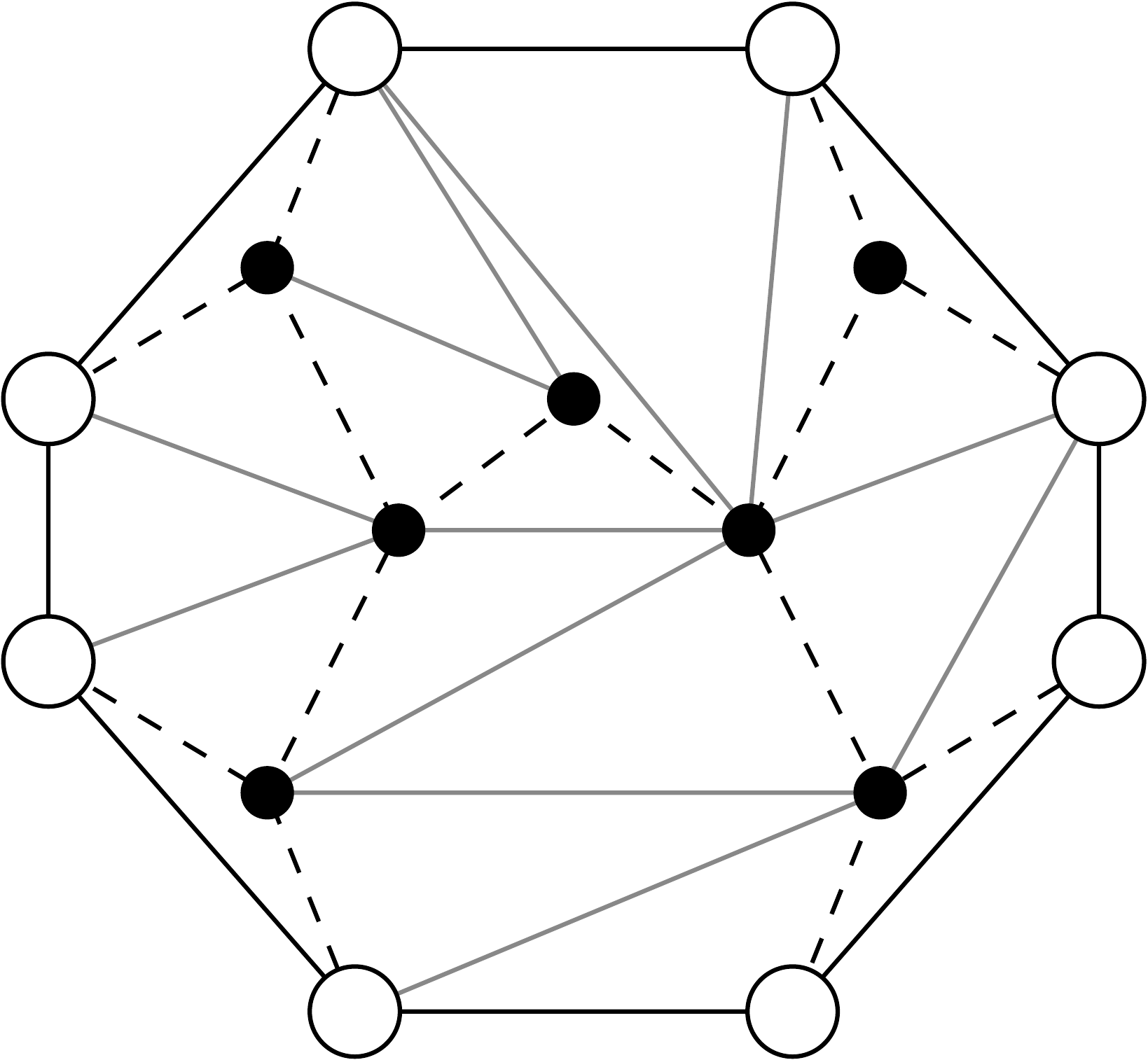}}
\caption{(a)~A face, (b)~augmentation with a balanced binary tree,
 (c)~triangulation with grey edges}
\label{fig:augment}
\end{figure}

Let $G$ be an $n$-vertex $k$-outerplanar graph with the maximum degree $\Delta$. If the
outerplanarity $k$ of $G$ is bounded by a constant, we can easily augment $G$ to logarithmic
diameter, preserving its constant maximum degree, as follows. Inside each non-triangular face $f$
of $G$, insert a balanced binary tree with
$\lceil \log |f|\rceil$ levels and $|f|$ leaves and then triangulate the
remaining non-triangular faces by inserting an \emph{outerpath}
(an outerplanar graph whose weak dual is a path) with constant maximum degree; see
\autoref{fig:augment}. However, such an augmentation results in a maximal
planar graph with the diameter $d=\Oh(k\log n)$, which does not yield a balanced
circle-contact representation when $k$ is non-constant. For $k=\Omega(\log n)$,
we present a different augmentation to achieve the diameter
$d=\Oh(k+\log n)$ in the resulting graph.

We augment the graph using \emph{weight-balanced binary trees}.
Let $T$ be a binary tree with
leaves $l_1, l_2, \ldots, l_{|f|}$
and a prespecified weight $w_i$ assigned to each leaf $l_i$.
The tree $T$ is \emph{weight-balanced} if the depth of each leaf $l_i$ in $T$ is $\Oh(\lceil \log (W/w_i) \rceil)$, where
$W=\sum_{i=1}^{f}w_i$. There exist several algorithms for producing a weight-balanced
binary tree with positive integer weights
defined on its leaves~\cite{GM59,NieRei-SJC-73}.

\begin{figure}[b!]
\center
\includegraphics[height=5.5cm]{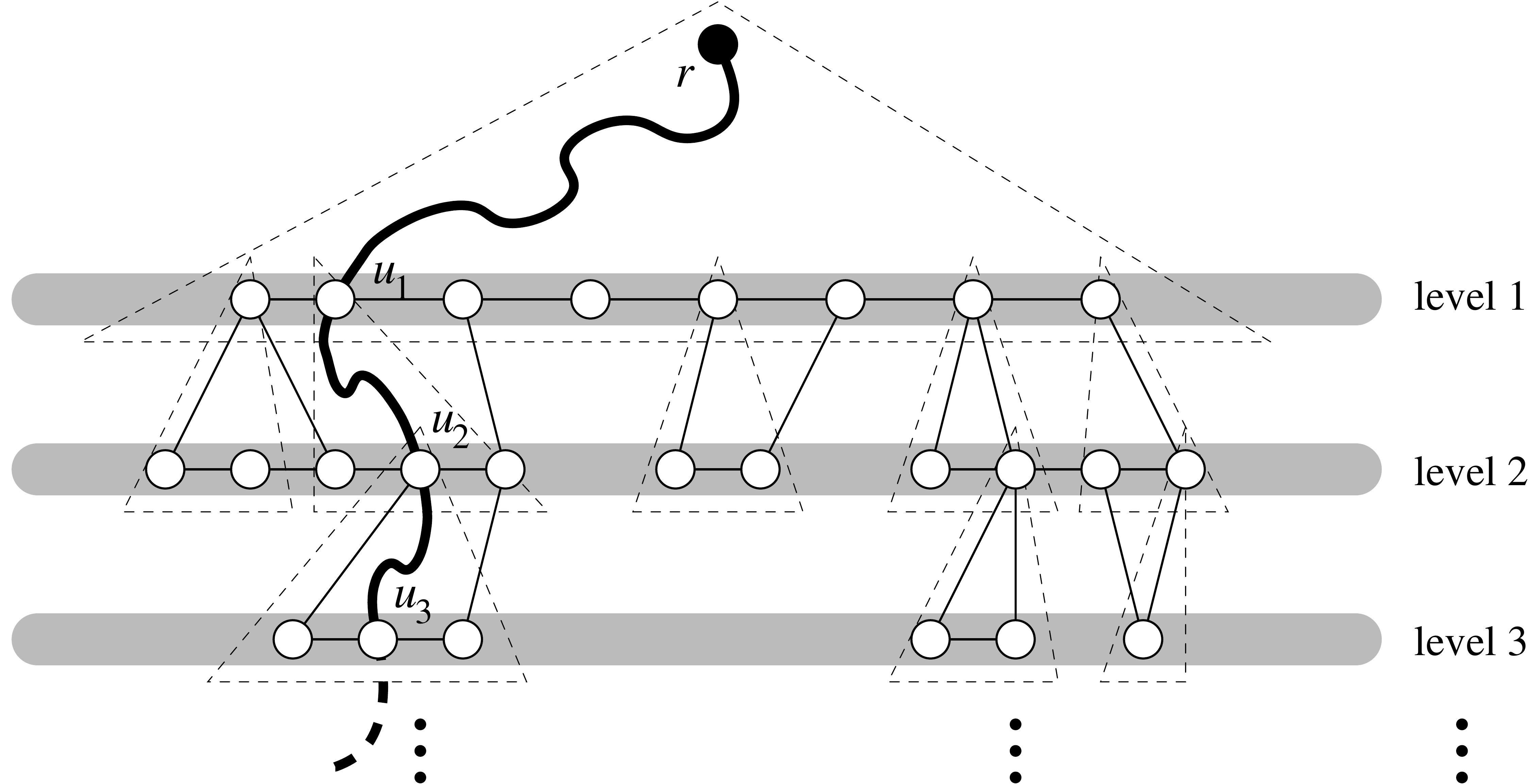}
\caption{Augmentation of $G$ with a weight-balanced binary trees}
\label{fig:weighted}
\end{figure}

To augment $G$, we label each vertex $v$ of $G$ with the number $l+1$,
where $l$ is the number of outer cycles that need to be removed before $v$ becomes
an outer vertex. By our assumption that the outerplanarity of $G$ is $k$, the label of every vertex is at most
$k$. It follows from this labeling that, for each vertex $v$ of $G$ with label $l>1$, there
exists a face $f$ containing $v$ such that $f$ has at least one vertex of label $l-1$ and such that all
the vertices on $f$ have label either $l$ or $l-1$. We insert a weight-balanced binary tree
inside $f$; we choose an arbitrary vertex of $f$ with label $l-1$ as the root of the tree, and a subset of vertices with
label $l$ as the leaves; see \autoref{fig:weighted}. We construct these trees inside the different
faces in such a way that each vertex of $G$ with label $l>1$ becomes a leaf in exactly one of
the trees.
Finally, we insert another weight-balanced tree $T_0$ on the outer face containing
all the outer vertices as the leaves.
Note that we have yet to specify the weights
we assign to these leaves for producing the weight-balanced trees.
By the construction, the union of all these trees forms a connected
spanning tree of $G$; we can consider the root of $T_0$ to be the root of the whole spanning tree.

Let us now specify the
weights assigned to the leaves of the different weight-balanced
trees. We label each tree with the label of its root, and define the weights for the leaves of each tree in a bottom-up ordering, by decreasing order of the labels of the trees.
In a tree $T$ with label $l=(k-1)$, all the leaves have label $k$ and are not the root
of any other tree; we assign each of these leave the weight~$1$. In this case, the total weight of $T$
is the number of its leaves. Similarly, for a tree with label
$l<k-1$, we assign a weight of $1$ to those leaves $v$ that do not have any tree rooted
at them; otherwise, if $v$ is the root of a tree $T_v$ with label $l+1$, the weight of $v$ is
the total weight of $T_v$. The total weight of $T$ is defined as the
summation of the weights of all its leaves.

Now, for each vertex $v$ of $G$, the distance to $v$ from the root $r$ of $T_0$
is $\Oh(k+\log n)$. Indeed, assume that $v=u_l$ is a vertex with label $l$ and $u_{l-1}$, $\ldots$,
$u_1$, $u_0=r$ are the root vertices of the successive weight-balanced trees $T_{u_{l-1}}$,
$\ldots$, $T_{u_1}$, $T_0$ with labels $l-1, \ldots, 1, 0$, respectively on the way from $v$
to $r$; see \autoref{fig:weighted}. Then the distance from $v$ to $r$ is
\ifFull
$$\Oh(\lceil \log w(r)/w(u_1) \rceil) + \Oh(\lceil \log w(u_1)/w(u_2) \rceil) + \ldots + 
\Oh(\lceil \log w(u_{l-1})/w(v) \rceil) = \Oh(k+\log w(r)).$$
\fi
Here $w(u_i)$ denotes the weight of vertex $u_i$ as the root;
$w(r)$ is the weight of the root of $T_0$, which is equal to the total number of
vertices, $n$, in $G$. Therefore, the diameter of the augmented graph is $\Oh(k+\log n)$, 
where the first term, $k$, comes from the ceilings in the summation.
Finally, we triangulate the graph by
inserting outerpaths with constant maximum degree inside each non-triangular face to
obtain a maximal planar graph with constant maximum degree and $\Oh(k+\log n)$ diameter.
The result follows from \autoref{lm:bounded-packing}.
\end{proofof}

\subsection{Negative Results}
Next we show that, for a graph with unbounded maximum degree or
unbounded outerplanarity, there might not be a balanced
circle-contact representation with circles.

\begin{lemma}
\label{lem:no-balanced} There is no balanced
circle-contact representation for the graphs in
 \autoref{fig:no-balanced}.
\end{lemma}

\begin{figure}[t!]
\centering
	\subfigure[]{
    \includegraphics[width=0.4\textwidth]{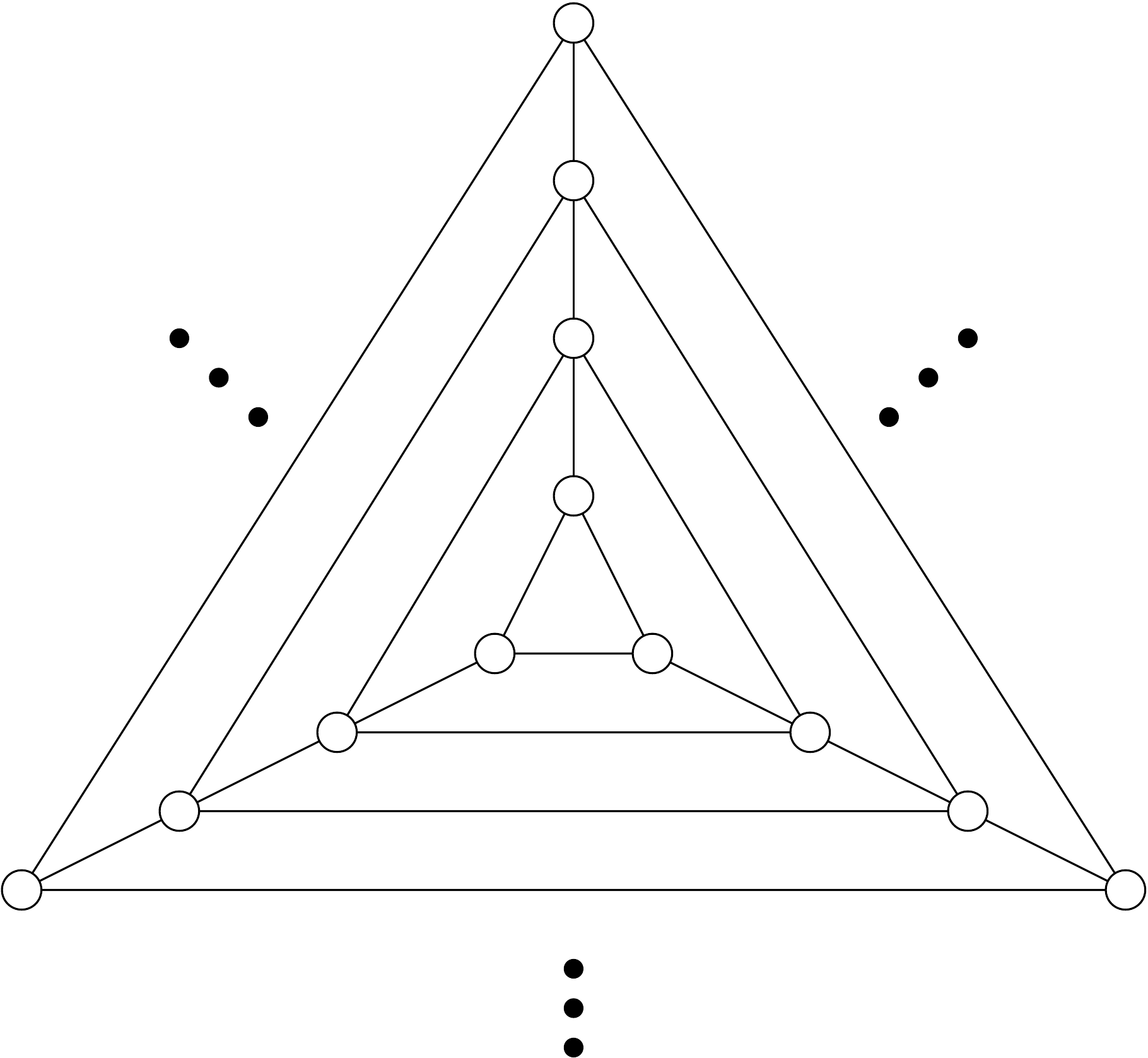}
    \label{fig:no-balanceda}}
    ~~~~~~~~~~~~~
	\subfigure[]{
    \includegraphics[width=0.4\textwidth]{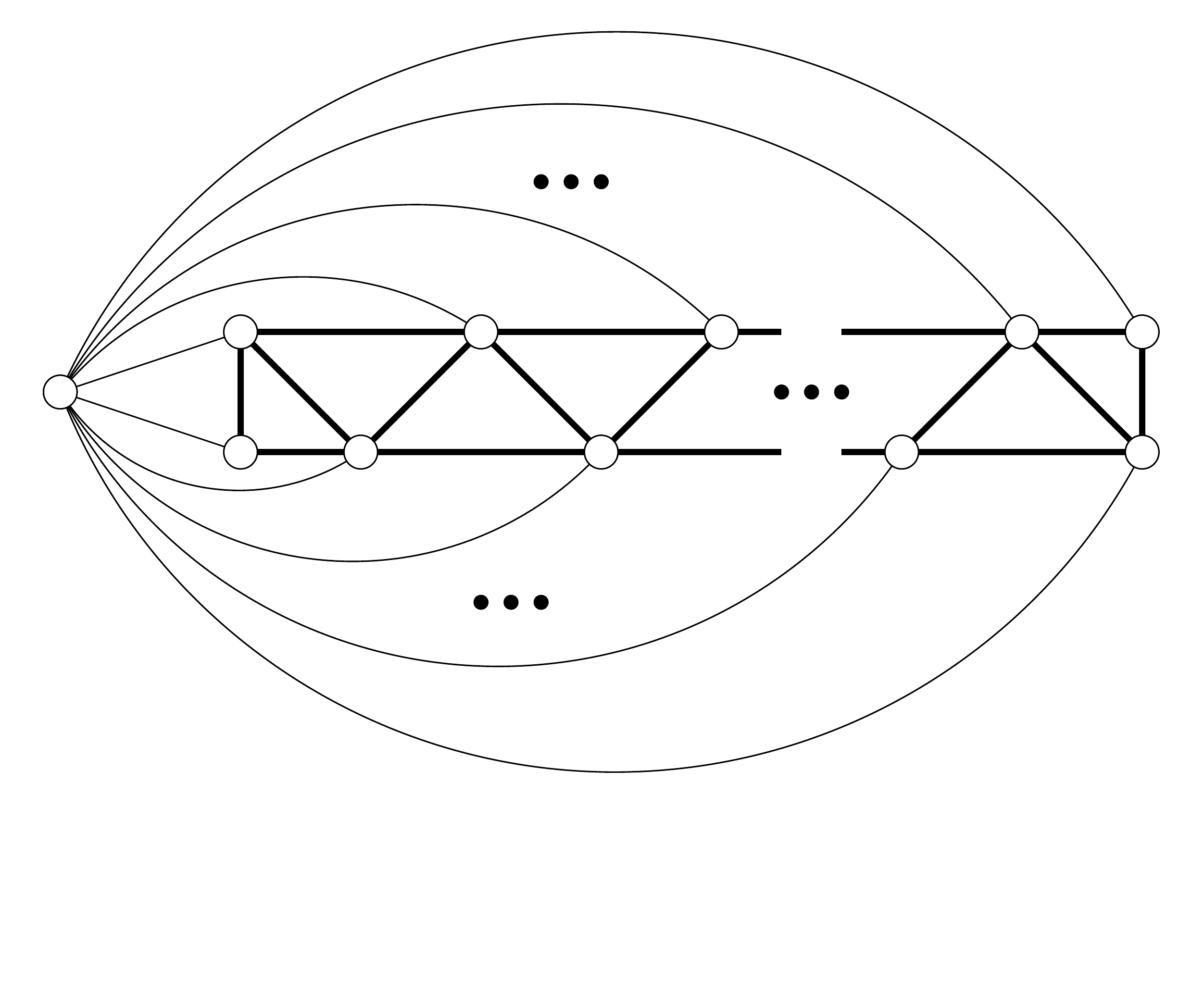}
    \label{fig:no-balancedb}}
\caption{Planar graphs with no balanced
circle-contact representation: (a)~the nested-triangles graph~\cite{DolLeiTri-ACR-84};
 (b)~a 2-outerplanar graph}
\label{fig:no-balanced}
\end{figure}

\ifFull

Recall that \autoref{lem:no-balanced}
states that there is no balanced
circle-contact representation for the graphs in
 \autoref{fig:no-balanced}.
We need the following auxiliary lemma to prove this.

\begin{figure}[htbp]
\centering
\includegraphics[width=0.25\textwidth]{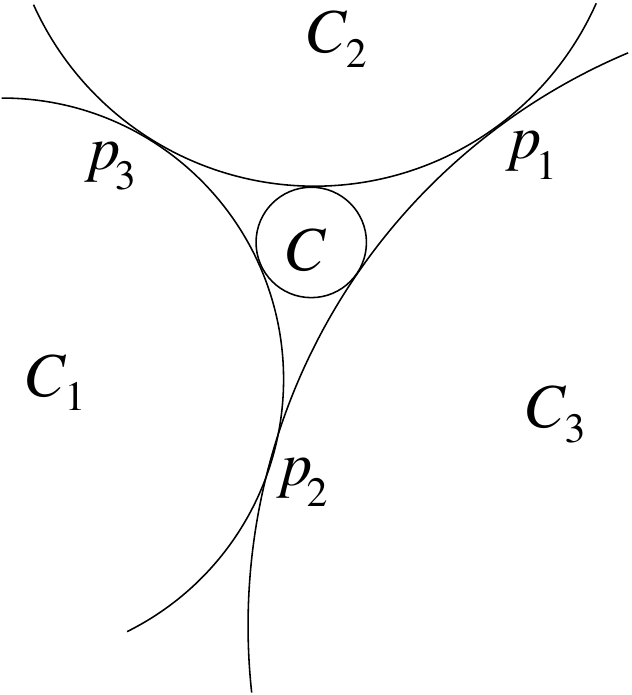}
\caption{Illustration for the proof of \autoref{lem:descartes}}
\label{fig:descartes}
\end{figure}

\begin{lemma}
\label{lem:descartes} Let $C_1$, $C_2$, $C_3$ be three interior-disjoint circles with radii $r_1$, $r_2$, $r_3$,
 respectively that touch each other mutually at the three points $p_1$, $p_2$ and
 $p_3$. Then every circle that lies completely inside the curvilinear triangle $p_1p_2p_3$ bounded by the three circles has radius at most half the median of $r_1$, $r_2$ and $r_3$.
\end{lemma}
\begin{proof} Assume without loss of generality that $r_1\le r_2\le r_3$, so that the median radius is~$r_2$. Let $C$ be a circle
 lying completely inside $p_1p_2p_3$ with radius $r$; see \autoref{fig:descartes}. The radius
 of $C$ is maximum when $C$ touches all the three circles $C_1$, $C_2$ and $C_3$. Indeed
 if $C$ does not touch at least one of the three circles, say $C_3$, then it can be moved slightly
 towards $C_3$ and then its radius can be increased by a positive amount so that it still lies inside
 the region. Thus, we can assume without loss of generality that $C$ touches all the three circles $C_1$, $C_2$, $C_3$.
 Then by Descartes' Circle Theorem~\cite{Descartes}, we have
 $$(\frac{1}{r_1}+\frac{1}{r_2}+\frac{1}{r_3}+\frac{1}{r})^2
 = 2((\frac{1}{r_1})^2+(\frac{1}{r_2})^2+(\frac{1}{r_3})^2+(\frac{1}{r})^2),$$
 which simplifies to
 $$\frac{1}{r} = \frac{1}{r_1} + \frac{1}{r_2} + \frac{1}{r_3}
 + 2\sqrt{\frac{1}{r_1r_2} + \frac{1}{r_2r_3} + \frac{1}{r_3r_1}}.$$
 (The other solution to this quadratic equation, with a negative sign on the square root term, gives the radius for the larger circle surrounding and tangent to $C_1$, $C_2$, and $C_3$.)
 Removing two positive terms from this equation gives the inequality $\frac{1}{r} \ge \frac{1}{r_1}+\frac{1}{r_2} \ge \frac{2}{r_2}$ or, equivalently, $r\le r_2/2$.
\end{proof}

\begin{proofof}{\autoref{lem:no-balanced}}
We prove the claim only for the 2-outerplanar
 graph in \autoref{fig:no-balancedb}, since the proof for the nested-triangles graph in
 \autoref{fig:no-balanceda} follows a similar argument.

 Because this graph is 3-connected, its planar embedding is unique up to the choice of the outer face. To begin with, assume that the outer face has been chosen as triangle $xa_1b_1$.
 Let $\Gamma$ be a contact
 representation for the graph with circles. We prove that $\Gamma$ is not balanced.
 Let $r$ be the radius of the circle $C$ representing $x$ and let $r_i$ and $r'_i$, $1\le i\le t$
 be the radii of the circles $C_i$ and $C'_i$ representing $a_i$ and $b_i$, respectively.
 Since $\Gamma$ is a contact representation, for any value $i\in\{1,\ldots t-1\}$,
 $C_i$ and $C'_i$ are completely inside the region created by three circles $C$, $C_{i+1}$
 and $C'_{i+1}$. Therefore by \autoref{lem:descartes}, the values of $r_i$ and $r'_i$
 are both less than $\frac{1}{2}\max(r_{i+1}, r'_{i+1})$. Since $t=(n-1)/2$, this implies that
 $\Gamma$ is not balanced.

 In the general case, in which the outer face can be an arbitrary triangle of the graph, the proof is similar. No matter which triangle is chosen as the outer face, the remaining subgraph between that face and one of the two ends of the zigzag chain of triangles in the graph has the same form as the whole graph, with at least half as many vertices, and the result follows in the same way for this subgraph.
\end{proofof}

\autoref{lem:no-balanced}
shows the tightness of the two conditions for balanced
\fi
circle-contact representations in \autoref{th:bounded}. Note that the example of the graph in
\autoref{fig:no-balancedb} can be extended for any specified maximum degree,
by adding a simple path to the high-degree vertex.
Furthermore,
the example is a $2$-outerplanar graph with no balanced
circle-contact representation.

\section{Trees and Outerplanar Graphs}

\ifFull
Here we address balanced circle-contact representation for trees
and outerplanar graphs.
We prove that every tree, every outerpath, and every cactus graph admits a balanced circle-contact
representation.
\fi

\ifFull
\subsection{Trees}
\fi

\begin{theorem}
\label{th:tree}
Every tree has a balanced circle-contact representation. Such a representation
can be found in linear time.
\end{theorem}

\begin{proof}
We first find a contact representation $\Gamma$ of a given tree $T$ with squares such that
 the ratio of the maximum and the minimum sizes for the squares is polynomial in the number
 of vertices $n$ in $T$. To this end, we consider $T$ as a rooted tree with an arbitrary vertex $r$
 as the root. Then we construct a contact representation of $T$ with squares where each
 vertex $v$ of $T$ is represented by a square $R(v)$ such that $R(v)$ touches the square for
 its parent by its top side and it touches all the squares for its children by its bottom side;
 see Figs.~\ref{fig:treea}~and~\ref{fig:treeb}. We choose the size of $R(v)$ as $l(v)+\eps (n(v)-1)$, where
 $\eps>0$ is a small positive constant and $n(v)$ and $l(v)$ denote the number of vertices
 and the number of leaves in the subtree of $T$ rooted at $v$. In particular, the
 size of $R(v)$ is 1 when $v$ is a leaf. If $v$ is not a leaf, then suppose $v_1$, $\ldots$,
 $v_d$ are the children of $v$ in the counterclockwise order around $v$. Then we place the
 squares $R(v_1)$, $\ldots$, $R(v_d)$ from left-to-right touching the bottom side of $R(v)$
 such that for each $i\in\{1,\ldots,d-1\}$, $R(v_{i+1})$ is placed $\eps$ unit to the right
 of $R(v_i)$; see \autoref{fig:treeb}. There is sufficient space to place all these squares
 in the bottom side of $R(v)$, since $n(v)=(\sum_{i=1}^dn(v_i))-1$ and
 $l(v)=\sum_{i=1}^dl(v_i)$. The representation contains no crossings or unwanted contacts
 since for each vertex $v$, the representation of the subtrees rooted at $v$ is bounded in the
 left and right side by the two sides of $R(v)$, and all the subtrees rooted at the children of $v$
 are in disjoint regions $\eps$ unit away from each other. The size of the smallest square
 is 1, while the size of the largest square (for the root) has size $l(T)+\eps (n-1) = \Oh(n)$,
 where $l(T)$ is the number of leaves in $T$.

\begin{figure}[t]
\centering
\subfigure[]
	{\includegraphics[width=0.33\textwidth]{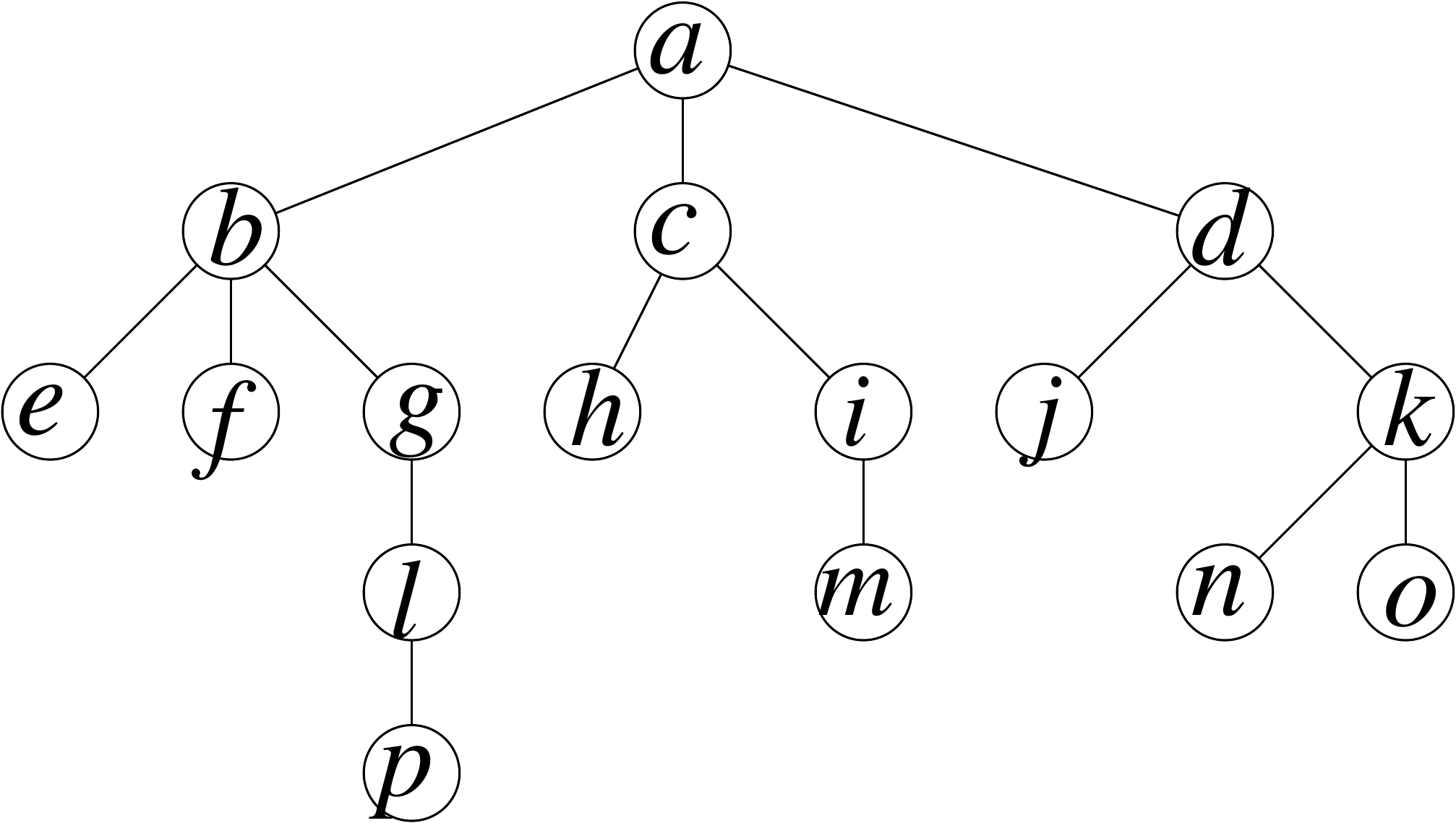}
    \label{fig:treea}}
\hfill
\subfigure[]
	{\includegraphics[width=0.3\textwidth]{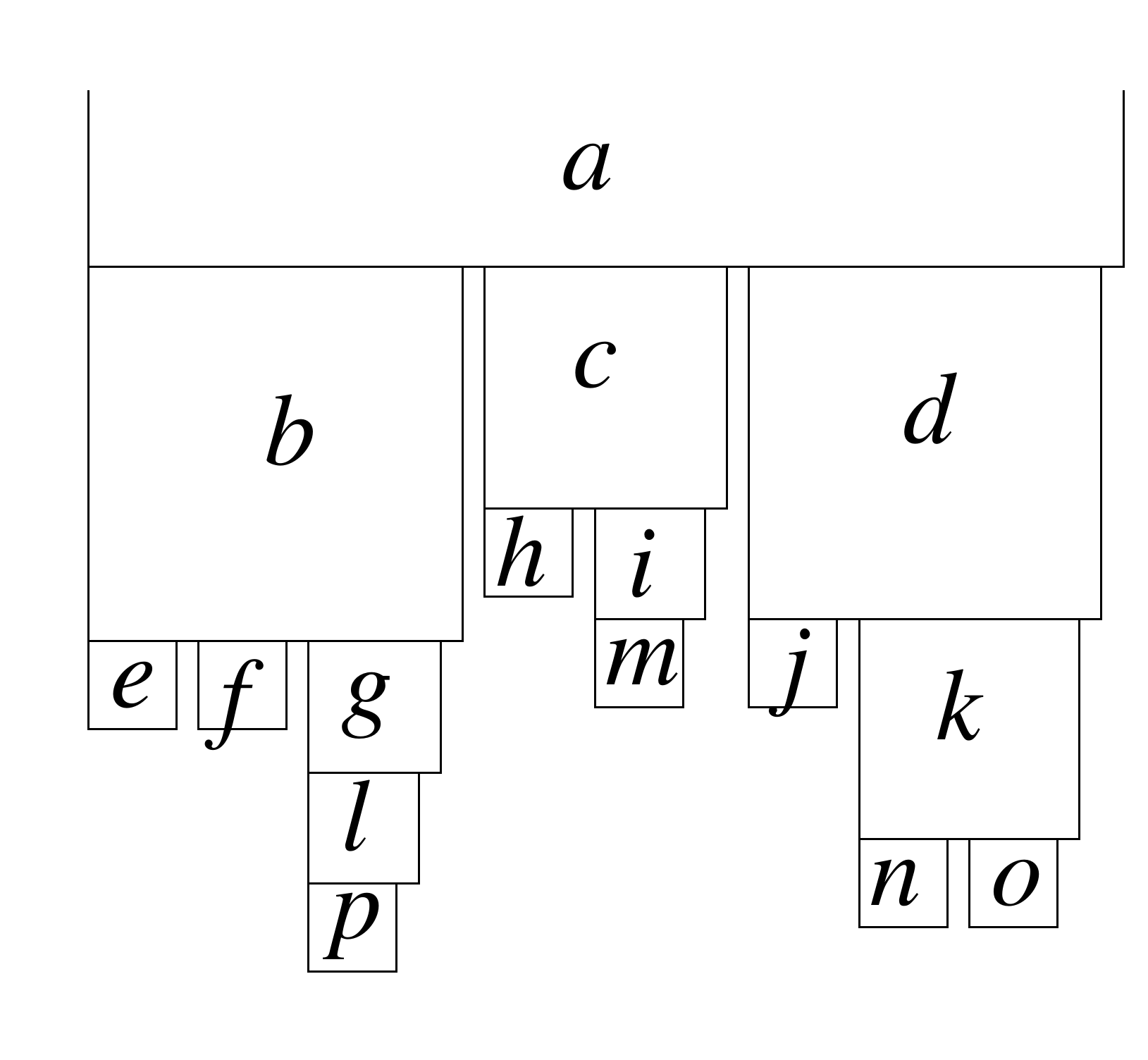}
    \label{fig:treeb}}
\hfill
\subfigure[]
	{\includegraphics[width=0.3\textwidth]{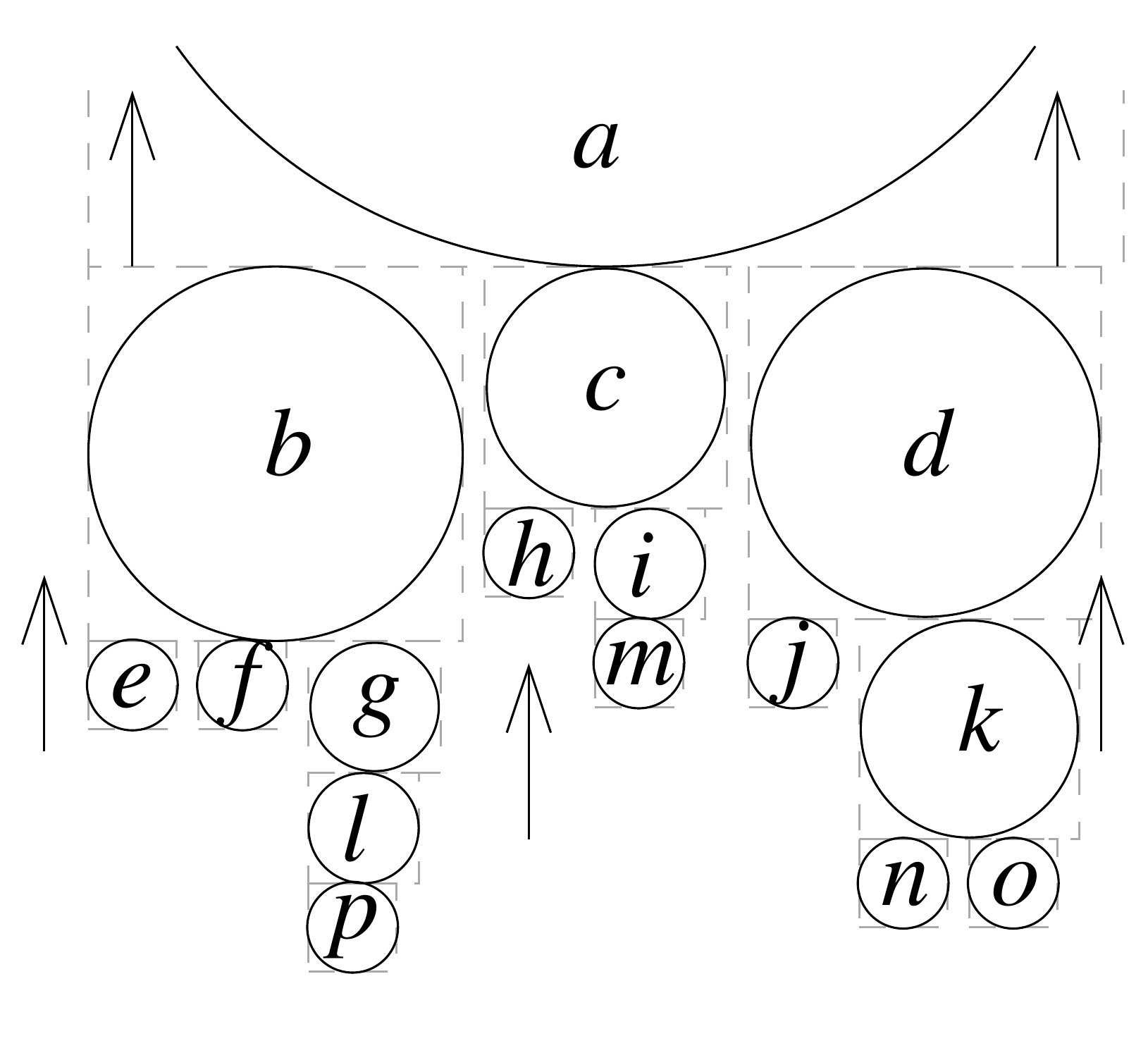}
    \label{fig:treec}}
\caption{Construction of a balanced circle-contact representation}
\label{fig:tree}
\end{figure}

Using $\Gamma$, we find a balanced circle-contact
 representation of $T$ as follows. We replace each square $R(v)$, representing vertex $v$,
 by an inscribed circle of $R(v)$; see \autoref{fig:treec}. The operation removes some contacts
 from the representation. We re-create these contacts by a top-down traversal
 of $T$ and moving each circle upward until it touches its parent. Note that a given circle
 will not touch or intersect any circle other than the circles for its parent and its children,
 as for every vertex in the infinite strip between its leftmost and rightmost point for
 its circle, the closest
 circle in the upward direction is its parent's one. Thus, we obtain a contact representation of
 $T$ with circles. The representation is balanced since the diameter for every circle is equal to
 the side-length for its square and we started with a balanced representation $\Gamma$.

The linear running time can be achieved by a linear-time traversal of $T$.
 First, by a bottom-up traversal of $T$, we compute the values $n(v)$ and $l(v)$ for each
 vertex $v$ of $T$. Using the values for each vertex, we compute the square-contact
 representation for $T$ by a linear-time top-down traversal of $T$. Finally, in another top-down
 traversal of $T$, for each vertex $v$ of $T$, we can compute the exact translation required
 for the inscribed circles of $R(v)$ to touch the parent circle.
\end{proof}

\ifFull
\subsection{Cactus Graphs}
\fi

Let us now describe how to compute a balanced circle-contact representation for a
\textit{cactus} graph, which is a connected graph in which every biconnected component
is either an edge or a cycle. We use the algorithm described in the proof of \autoref{th:tree},
and we call it \textbf{Draw\_Tree}.

Let $T$ be a rooted tree with a plane embedding. For each vertex $v$ of  $T$, add an edge
 between every pair of the children of $v$ that are consecutive in the clockwise order around
 $v$. Call the resulting graph an \textit{augmented fan-tree} for $T$. Clearly for any rooted
 tree $T$, the augmented fan-tree is outerplanar. We call an outerplanar graph a
 \textit{fan-tree} graph if it is an augmented fan-tree for some rooted tree. A \textit{star}
 is the complete bipartite graph $K_{1,n-1}$. The center of a star is the vertex that is adjacent
 to every other vertex. An augmented fan-tree for a star is obtained by taking the center
 as the root. Thus, an augmented fan-tree for a star is a \textit{fan}. The \textit{center} of
 a fan is again the vertex adjacent to all the other vertices.

\begin{lemma}
\label{lem:fan}
Every subgraph of a fan admits a contact representation with circles in which,
for each circle $c(v)$ representing a vertex $v$ other than the center, the vertical strip
containing $c(v)$ is empty above $c(v)$.
\end{lemma}

\begin{proof} Let $G$ be a subgraph of a fan and let $T$ be the star contained in the fan. We
 now use the contact representation $\Gamma$ of $T$ obtained by \textbf{Draw\_Tree}
 to compute a representation for $G$. Consider the square-contact representation computed
 for $T$ in the algorithm. This defines a vertical strip for each circle $c(v)$ in $\Gamma$
 representing a vertex $v$, and for all the vertices other than the center, these strips are
 disjoint; see \autoref{fig:subfana}. Call the left and right boundary of this strip the
 \textit{left-} and \textit{right-line} for $c$, respectively.

\begin{figure}[tb]
\centering
\subfigure[]{
	\includegraphics[width=0.45\textwidth]{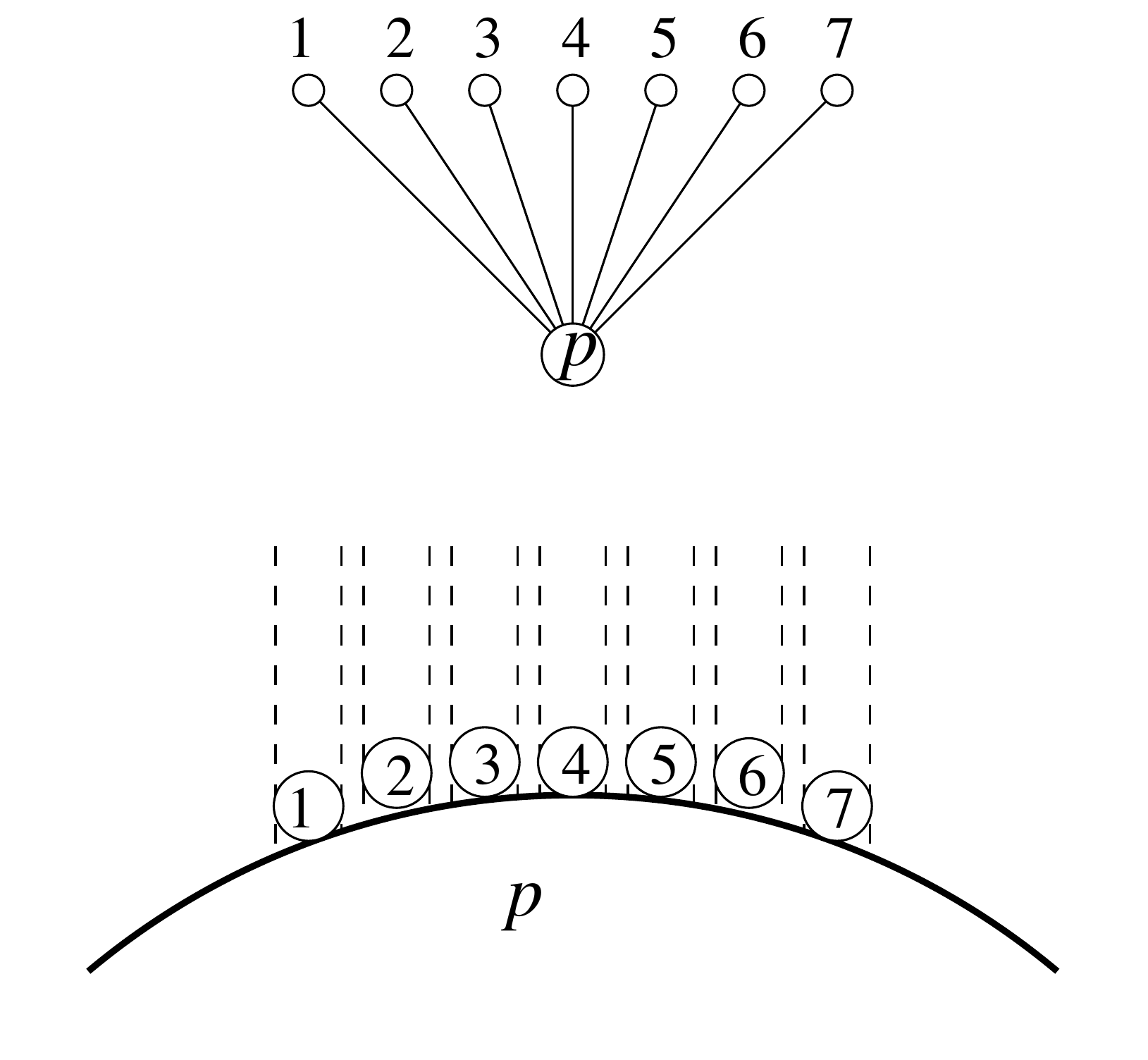}
    \label{fig:subfana}}
~~~~~~~~
\subfigure[]{
	\includegraphics[width=0.45\textwidth]{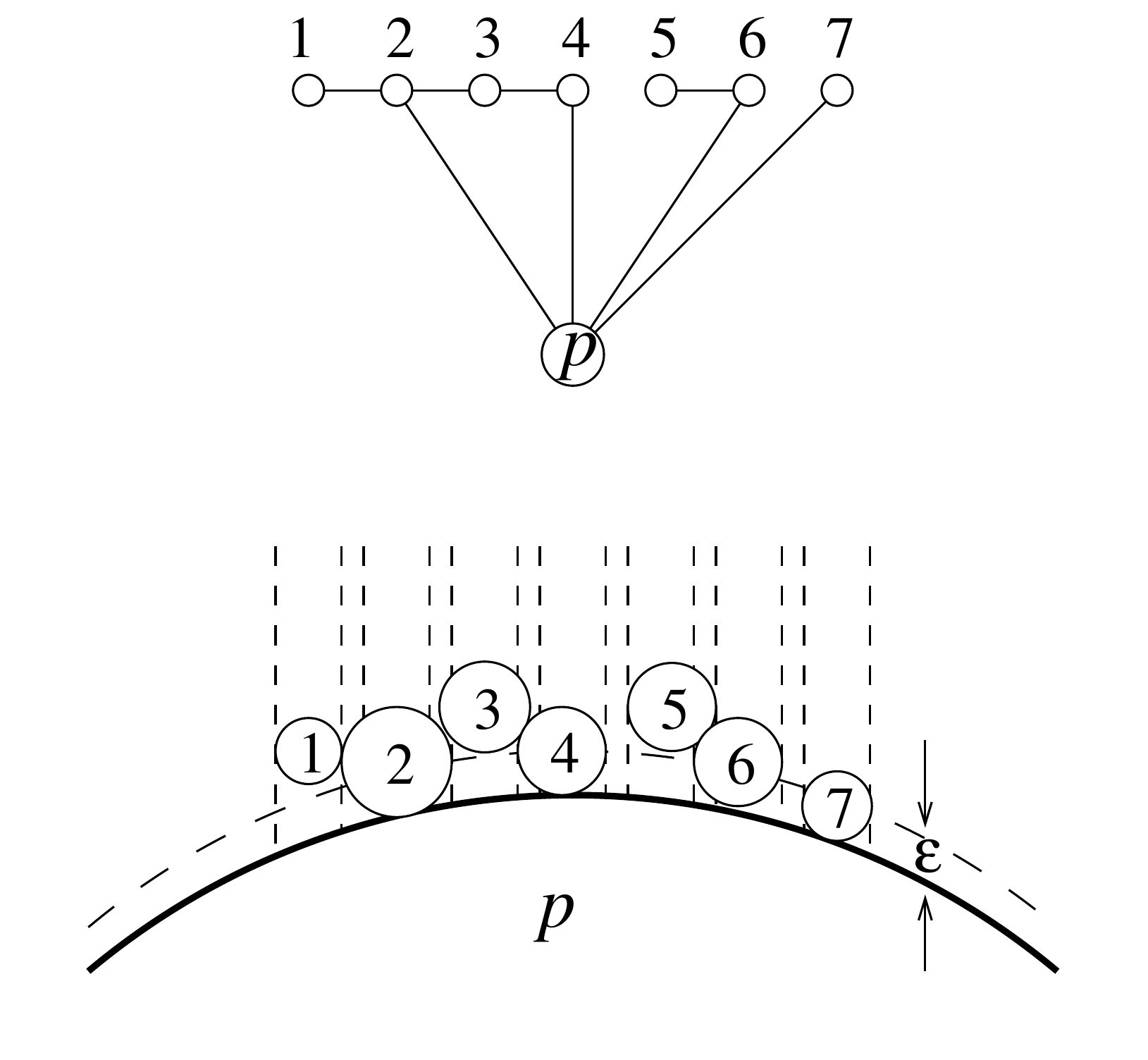}
    \label{fig:subfanb}}
\caption{(a)~A star $T$ and a contact representation of $T$ with circles;
	(b)~a subgraph of the fan for $T$ and its contact representation with circles}
\end{figure}

We now consider a set $S$ of circles, one for each vertex of $G$ other than the center,
with the following properties:

\begin{enumerate}[(P1)]
	\item The circles are interior-disjoint.

	\item Each circle $c'(v)$ representing a vertex $v$ spans the entire width of the vertical
		strip for $v$, and the vertical strip above $c'(v)$ is empty.

	\item For each vertex $v$, the circle $c'(v)$ touches the circle $c_0$ representing the center in
		$\Gamma$ if $v$ is adjacent to the center; otherwise, $c'(v)$ is exactly $\eps$
		distance away from $c_0$, for some fixed constant $\eps > 0$.

	\item If a vertex $v$ is not adjacent to the vertex on its left (or if $v$ is the leftmost vertex),
		then the leftmost point of $c'(v)$ is on the left-line of $v$;
              similarly, if $v$ is not
		adjacent to the vertex on its right (or if $v$ is the rightmost vertex), then the rightmost
		point of $c'(v)$ is on the right-line of $v$.

	\item The sizes for the circles are maximal with respect to the above properties.
\end{enumerate}

Note that there exists a set of circles with the properties (P1)--(P4); in particular, the set of
 circles in $\Gamma$ representing the vertices of $T$ other than the center is such a set.
We now claim that the set $S$ of circles with properties (P1)--(P5) together with the circle
 $c_0$ gives a contact representation for $G$; see \autoref{fig:subfanb}. First note that
 a circle $c'(v)$ cannot touch any circle other than $c_0$ and the two circles $c(v_l)$ and
 $c(v_r)$ representing the vertices $v_l$ and $v_r$ on its left and right, respectively. Indeed,
 it cannot pass the vertical strip for $v_l$ and $v_r$ above them due to (P2) and behind them
 due to (P3). Furthermore, the $\eps$ distance between $c_0$ and the circles for vertices
 non-adjacent to the center and the restriction on the left and right side in (P4) ensures that
 there is no extra contact. Hence, it is sufficient to show that for each edge in $G$, we have the
 contact between the corresponding circles.

Since each circle $c'(v)$ is maximal in size, it must touch at least three objects. One of them is
 either the circle $c_0$ or the $\eps$ offset line for $c_0$. Thus, if $v_l$ and $v_r$ are the
 left and right neighbors of $v$ (if any), then $c'(v)$ must touch two of the followings:
 (i)~$c'(v_l)$ (or the left line of $v$ if $v_l$ does not exists), (ii)~the right line for $v_l$,
 (iii)~$c'(v_r)$ (or the right line of $v$ if $v_r$ does not exists), and (iv) the left line for $v_r$.
 Assume without loss of generality that both $v_l$ and $v_r$ exist for $v$. Then if $c'(v)$
 touches both $c'(v_l)$ and $c'(v_r)$, we have the desired contacts for $v$.
Therefore,
 for a desired contact of $c'(v)$ to be absent,
 either $c'(v)$ touches both $c'(v_l)$ and the
 right-line of $v_l$ (and misses the contact with $c'(v_r)$), or it touches both $c'(v_r)$ and
 the left-line of $v_r$ (and misses the contact with $c'(v_r)$).

Assume, for the sake of a contradiction,
 that there are two consecutive vertices $x$ and $y$ that are
 adjacent in $G$ but $c'(x)$ and $c'(y)$ do not touch each other. Let $l$ and $r$ be the
 vertices to the left of $x$ and to the right of $y$, respectively. Then it must be the case that $x$ touches
 both $c'(l)$ and the right line for $l$ and $y$ touches both $c'(r)$ and the left line of $r$;
 see \autoref{fig:maximala}. One can then increase the size of either $c'(x)$ or $c'(y)$
 (say $c'(y)$) such that it now touches $c'(x)$ and the left-line for $r$ (but not $c'(r)$),
 a contradiction to the maximality for the circles; see \autoref{fig:maximalb}.
\end{proof}

\begin{figure}[t!]
\centering
\subfigure[]{
	\includegraphics[width=0.47\textwidth]{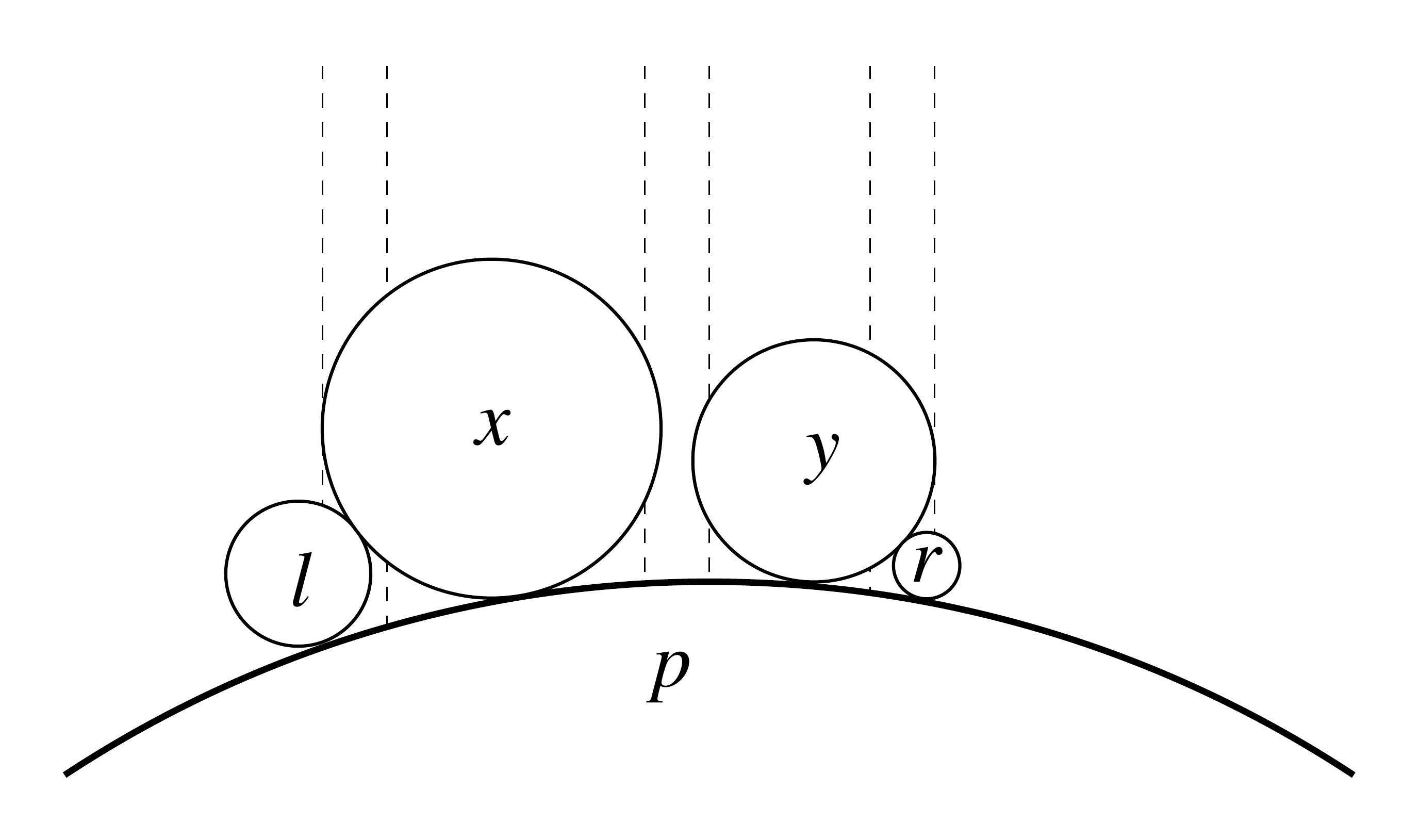}
    \label{fig:maximala}}
\hfill
\subfigure[]{
	\includegraphics[width=0.47\textwidth]{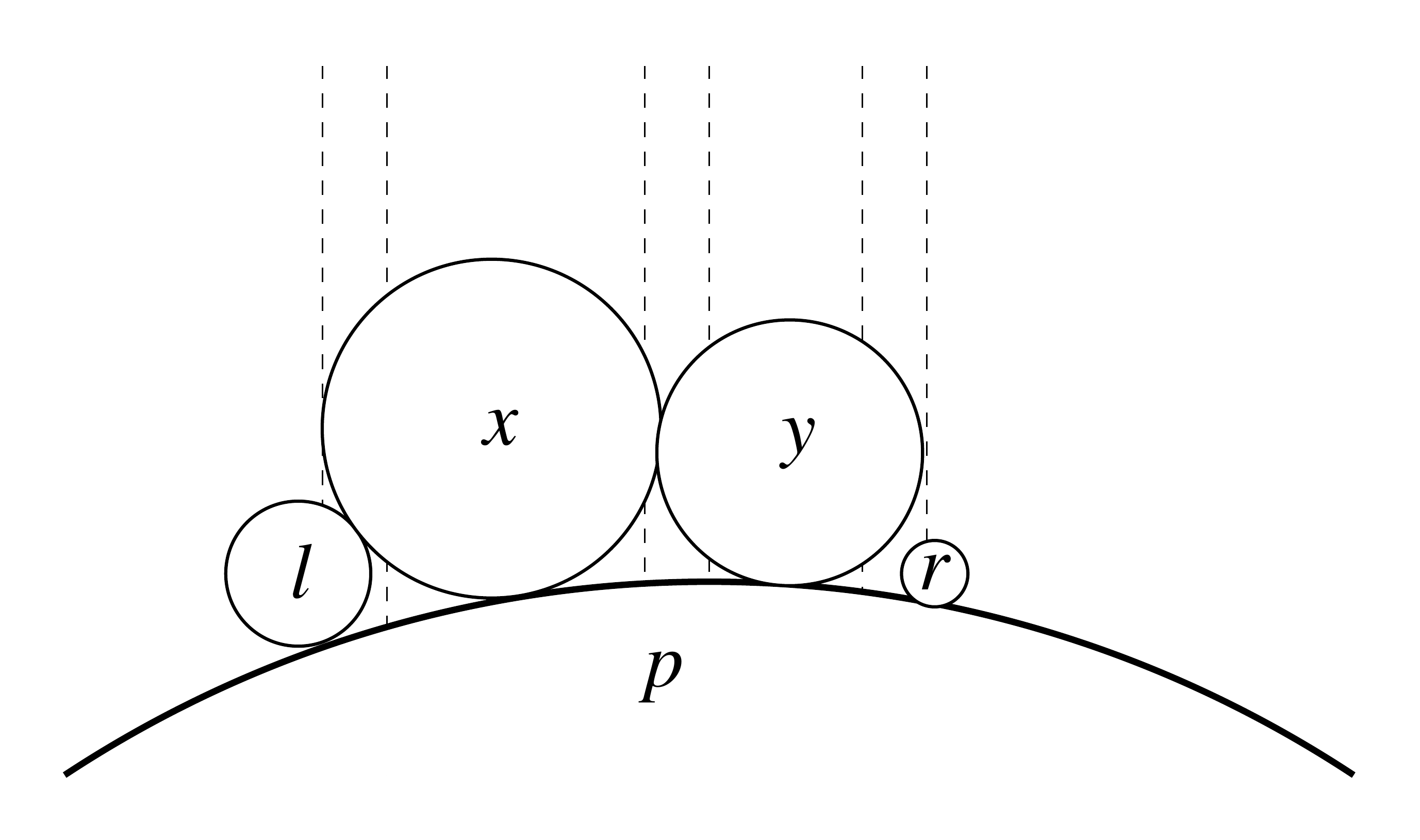}
    \label{fig:maximalb}}
\caption{Illustration for the proof of \autoref{lem:fan}: if the circles for $x$ and $y$ do not touch each other,
 at least one can be increased in size}
\end{figure}

Using the lemma, we can obtain a quadratic-time algorithm as follows. Given a subgraph $G$
of a fan, compute the balanced circle-contact representation $\Gamma$ for the corresponding star $T$ using
 \textbf{Draw\_Tree}. Then pick the vertices of $T$ other than the center in an arbitrary order and
 for each vertex $v$, replace the circle $c(v)$ in $\Gamma$ by a circle of maximum size
 that does not violate any of the properties (P1)--(P4) in the proof of \autoref{lem:fan}.
 This takes a linear time. Now for every edge $(x,y)$ for which $c(x)$ and $c(y)$ do not touch,
 replace one of the two circles (say, $c(y)$) with a circle that touches $c(x)$ as in
 \autoref{fig:maximalb}. Note that this may result in a loss of a contact between $c(y)$ and the
 circle to its right. We perform a similar operation for the circle to the right of $c(y)$, then
 possibly for the circle on its right and so on, until all missing contact are repaired. This
 process requires linear time per edge; hence, the total running time to compute the desired
 contact representation is quadratic. The contact representation is
 balanced since the representation obtained by \textbf{Draw\_Tree} is balanced
 and afterwards we only increase the size of circles that are not of the largest size.

\begin{theorem}
\label{th:fan-tree}
Every $n$-vertex fan-tree graph has a balanced circle-contact representation. Such a representation
can be found in $\Oh(n^2)$ time.
\end{theorem}

\begin{proof} Let $G$ be a fan-tree graph and let $T$ be the corresponding tree for which
 $G$ is the augmented fan-tree. Using \textbf{Draw\_Tree}, we first obtain a
 balanced circle-contact representation of $T$. As in the proof of \autoref{lem:fan},
 this defines a vertical strip for each vertex in $T$. In a top-down traversal of $T$, we
 can find a contact representation of $G$ with circles by repeating the
 quadratic-time algorithm for the subgraphs of fans. Hence, the total complexity is
 $\displaystyle\sum_{v\in V(T)}\deg_T^2(v) = \Oh(n^2)$.
\end{proof}

As a corollary of \autoref{th:fan-tree}, we obtain an algorithm
for creating balanced circle-contact representation of a cactus graph.

\begin{corollary}
Every $n$-vertex cactus graph has a balanced circle-contact representation. Such a representation
can be found in $\Oh(n^2)$ time.
\end{corollary}

\begin{proof}
Given cactus graph $G$, choose a root vertex $v$ arbitrarily. For each cycle $C$ of~$G$, add an edge from each vertex of $C$ to the (unique) closest vertex to $v$ in~$C$ (\autoref{fig:cactus}). The resulting supergraph of $G$ is a fan-tree; the result follows by \autoref{th:fan-tree}.
\end{proof}

\begin{figure}[tb]
\centering
\subfigure[]
	{\includegraphics[width=0.35\textwidth]{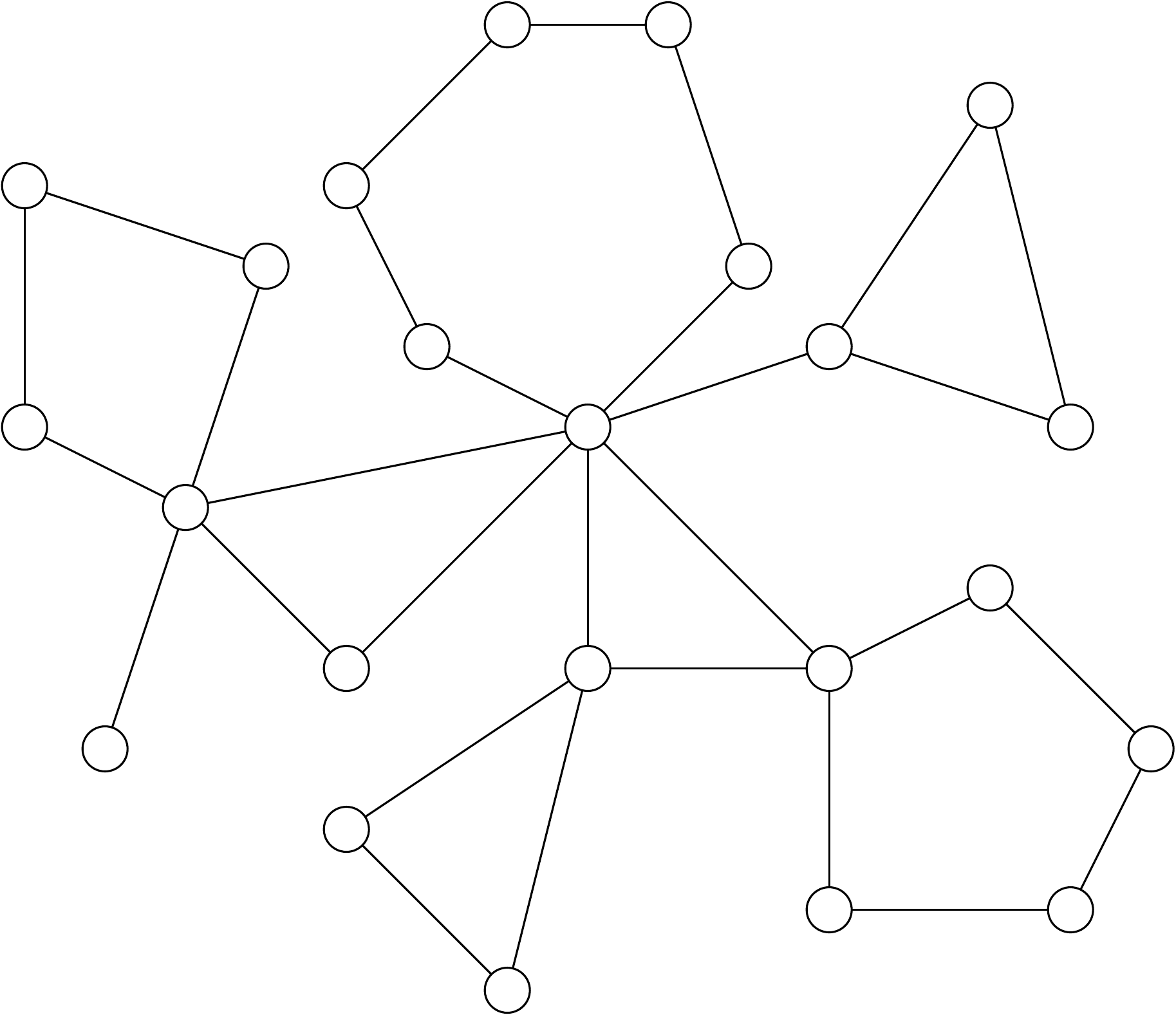}}
~~~~~~~~~~~~~~~~~~
\subfigure[]
	{\includegraphics[width=0.35\textwidth]{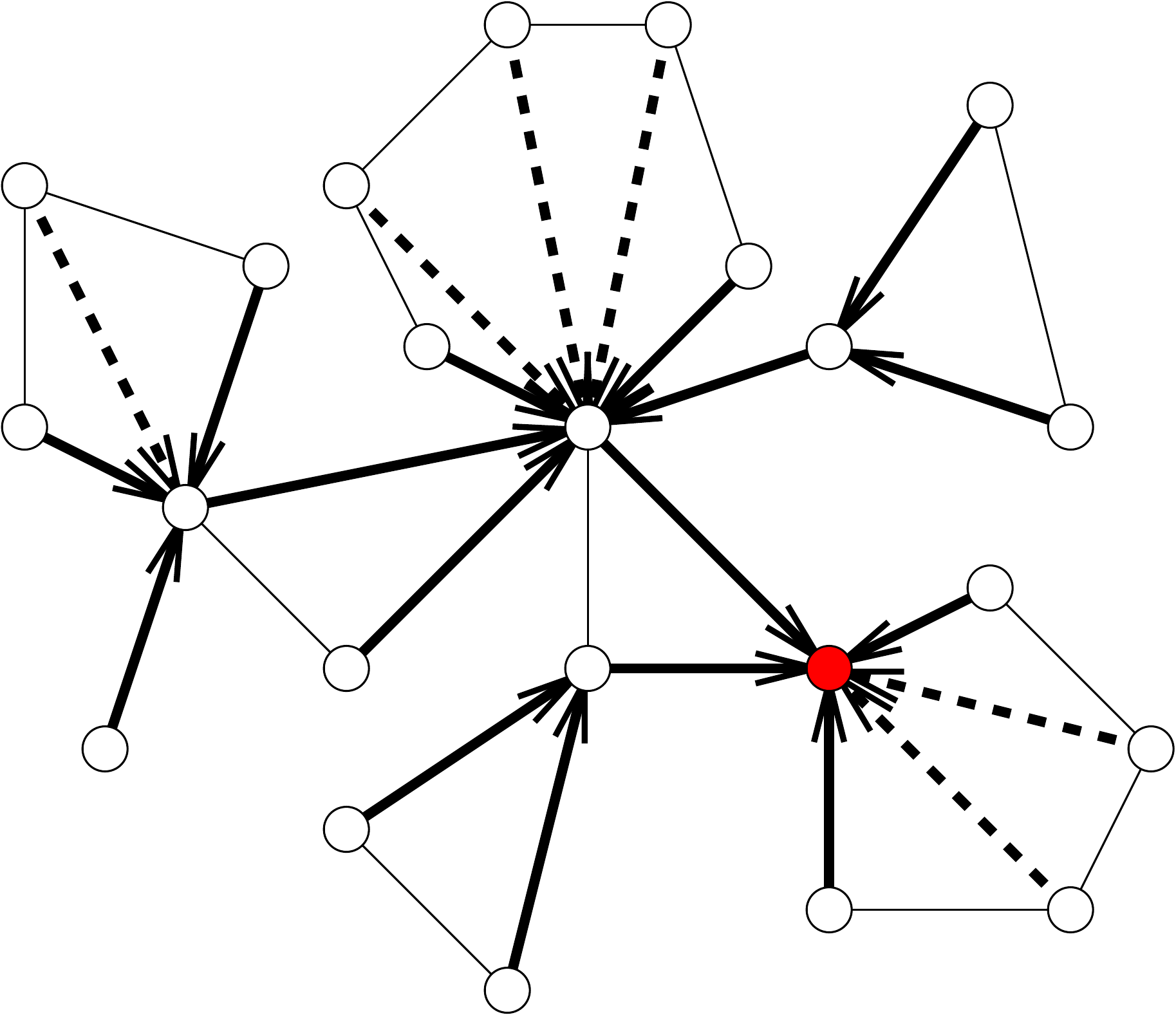}}
\caption{(a)~A cactus graph $G$; (b)~augmenting $G$ to a fan-tree so that the directed edges
 form a rooted tree and are oriented towards the root}
\label{fig:cactus}
\end{figure}

\ifFull

\subsection{Outerpaths}

We end this section with a linear-time algorithm for balanced circle-contact representation of
outerpaths. Recall that \textit{outerpaths} are outerplanar graphs whose weak dual is a path.

\begin{theorem}
\label{th:outerpath}
Every outerpath has a balanced circle-contact representation. Such a representation
can be found in linear time.
\end{theorem}

\begin{proof} Let $G$ be an outerpath. Assume for now that all the faces of $G$ are triangles.
 Then the consecutive faces of $G$ along the weak dual path can be partitioned into a sequence of fans. Define a path $P$ in
 $G$ that consists of all the edges that are shared by two maximal fans in $G$, along with the two
 terminal edges; see \autoref{fig:outerpath-a}. We will find a contact representation of $G$ in which
 each vertex of $P$ is realized by a unit circle. Number the vertices from
 one end to the other along $P$ as $1,2,\ldots, p$. Then for each vertex $i$ on $P$, draw
 a unit circle $C_i$ with the center at $(i,y)$, where $y=0$ if $i$ is even, and otherwise
 $y=\sqrt{3}/2$; see~\autoref{fig:outerpath-b}.

\begin{figure}[htbp]
\centering
\subfigure[]
	{\includegraphics[width=0.62\textwidth]{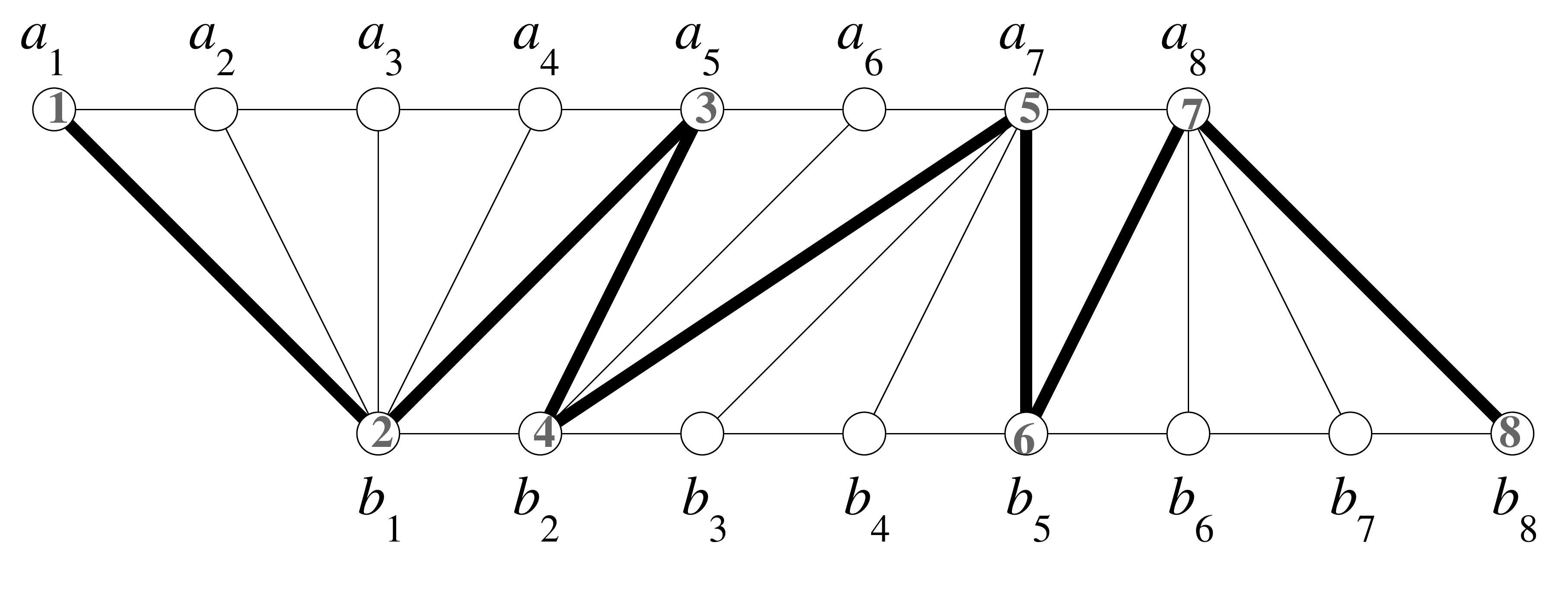}
	\label{fig:outerpath-a}}
\subfigure[]
	{\includegraphics[width=0.47\textwidth]{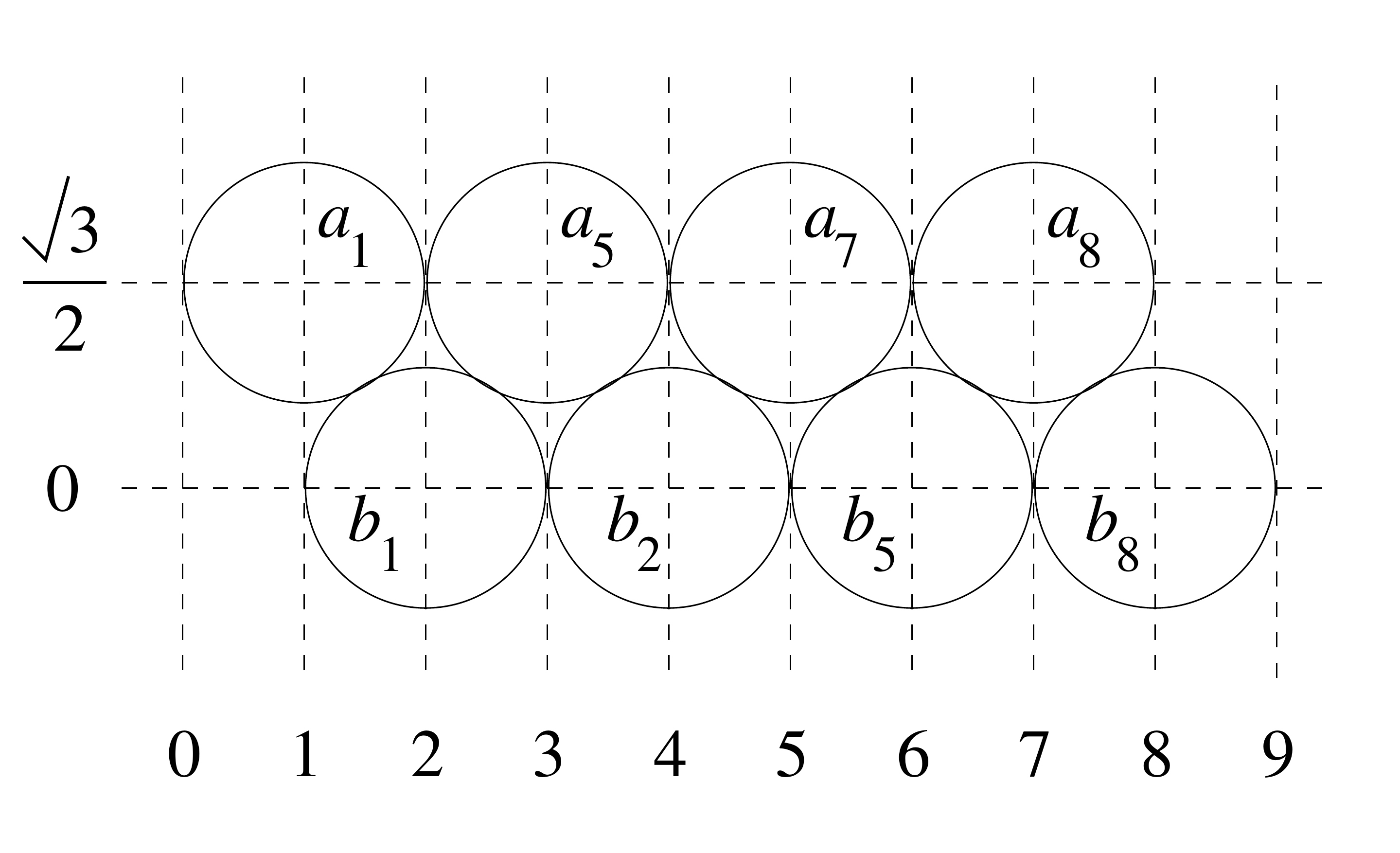}
	\label{fig:outerpath-b}}
\subfigure[]
	{\includegraphics[width=0.47\textwidth]{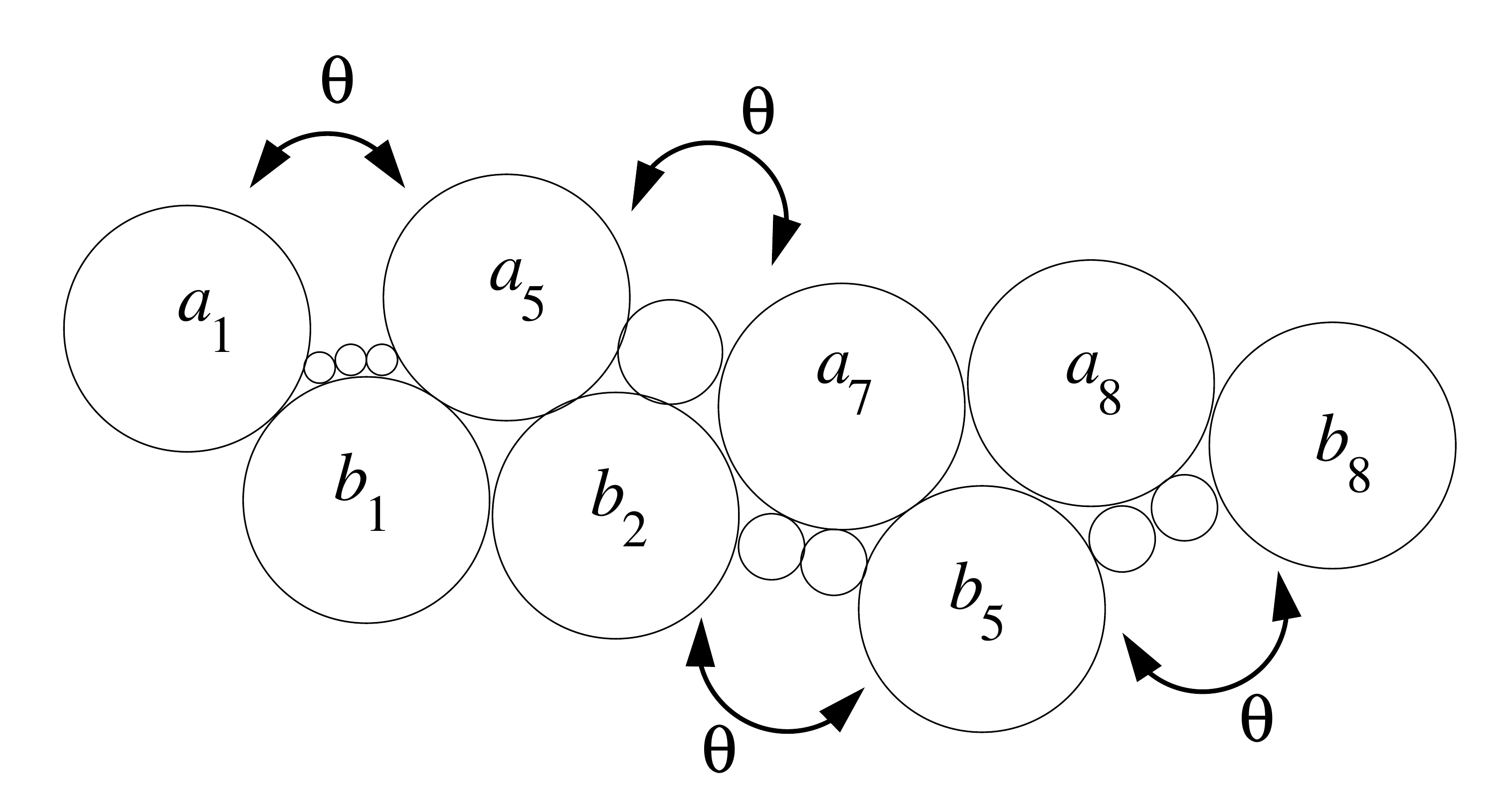}
	\label{fig:outerpath-c}}
\caption{(a)~An outerpath $G$; (b)--(c) construction of a balanced circle-contact representation for~$G$}
\end{figure}

 In this drawing, each circle $C_i$ touches the circles $C_{i-2}$, $C_{i-1}$, $C_{i+1}$ and $C_{i+2}$.
 Thus we have some additional and unwanted tangencies in this realization: namely, for some circles $C_i$, the
 tangency with $C_{i+2}$ (or symmetrically with $C_{i-1}$) does not correspond to an edge in $G$. When $i$ is odd, we may remove
an extra tangency between circles $C_i$  and $C_{i+2}$ by rotating $C_{i+2}$, and all circles to the left of it, clockwise by a small angle $\theta$ around the center of circle $C_{i+1}$. Similarly, when $i$ is even, we may remove an extra tangency between circles $C_i$ and $C_{i+2}$ by rotating $C_{i+2}$, and all circles to the left of it, counterclockwise by a small angle $\theta$ around the center of $C_{i+1}$. We choose $\theta$ to be $90^\circ/n$ so that these rotations
 do not create any overlap between the other circles.

 Next, we draw the circles representing the remaining vertices of each fan, within the regions created
 by these rotations; see \autoref{fig:outerpath-c}.
 It is easy to see that the sizes for these circles are at least $\Omega(1/n^2)$.
 Finally, if $G$ has some non-triangular faces, then
 we can add additional edges to augment $G$ into triangulated outerpath $G'$, and perform the construction above. A small perturbation of the construction suffices to remove its extra adjacencies. Clearly, the running time is linear.
\end{proof} 
\fi

\section{Bounded Tree-Depth}
\label{sec:tree-depth}

\ifFull

The \emph{tree-depth} of a graph~\cite{NesOss-12} is a measure of its complexity, related to but weaker than its treewidth. Thus, problems that are hard even for graphs of bounded treewidth, such as testing 1-planarity, can be solved efficiently for graphs of bounded tree-depth~\cite{BanCabEpp-WADS-13}, and constructions that might be impossible for graphs of bounded treewidth, such as finding balanced circle-contact representations even when the vertex degree is unbounded, might become possible for bounded tree-depth.

One definition of tree-depth involves a certain kind of forest, which we call an \emph{ancestral forest} for a given graph $G$; see~\autoref{fig:treedepth-a}. An ancestral forest is a forest $F$ having the vertices of $G$ as its vertices, such that every edge of $G$ connects a vertex $v$ with one of the ancestors of $v$ in the forest. For instance, every depth-first-search forest is ancestral; however, unlike depth-first search forests, an ancestral forest is allowed to have edges that do not belong to $G$. The tree-depth of $G$ is then the minimum height of an ancestral forest for~$G$, where the height of a tree is the maximum number of vertices on a root-leaf path. Tree-depth can also be defined inductively (a connected graph has tree-depth $\le k$ if it is possible to remove a vertex so that each remaining connected component has tree-depth $\le k-1$) or alternatively it can be defined in terms of the existence of certain graph colorings, called \emph{centered colorings}~\cite{NesOss-12}.

\begin{figure}[h!]
\centering
\subfigure[]
	{\includegraphics[height=4cm]{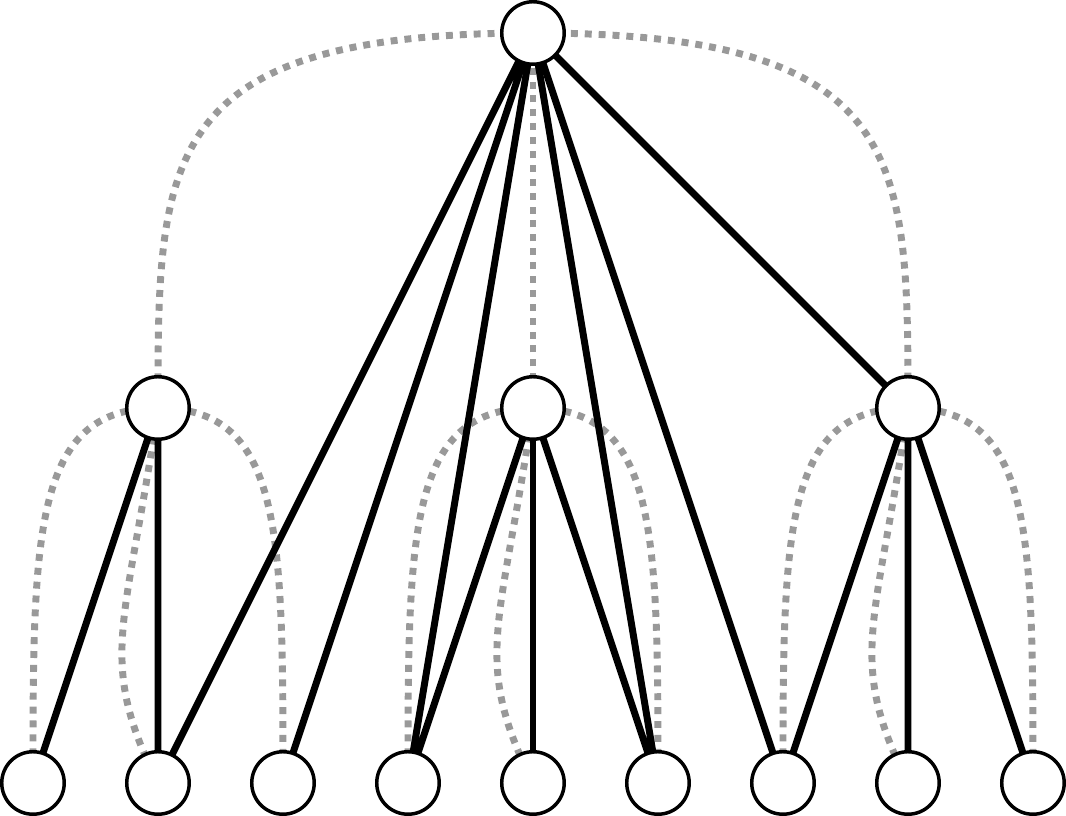}
	\label{fig:treedepth-a}}
~~~~~~~~~~~~~~~
\subfigure[]
	{\includegraphics[height=4cm]{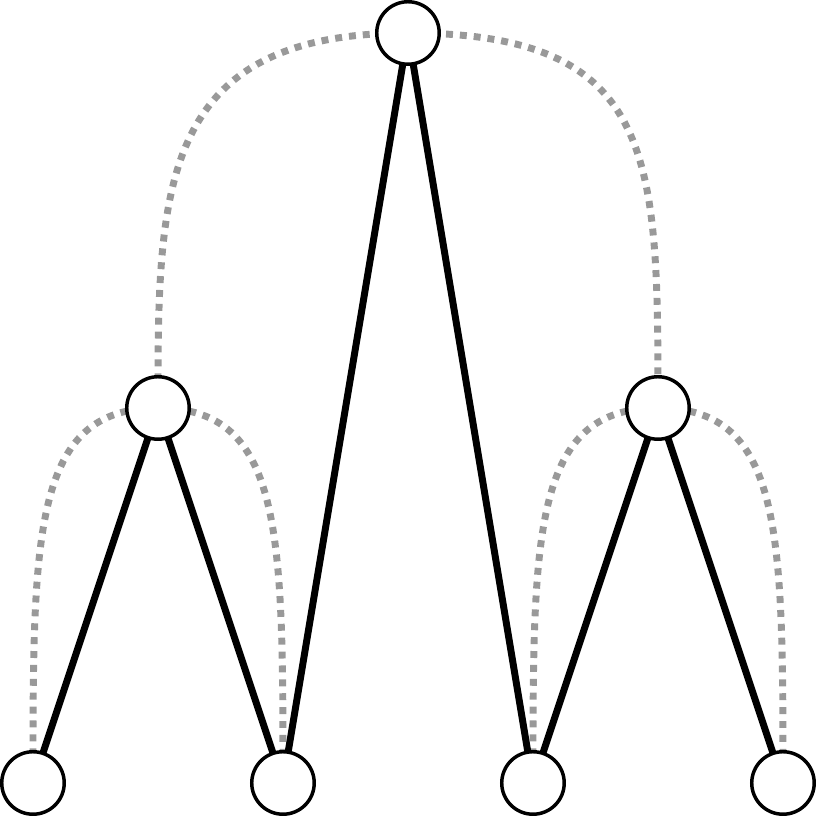}
	\label{fig:treedepth-b}}
\caption{(a)~A planar graph of tree-depth~$3$ and one of its possible ancestral trees, (b)~the longest path in a graph of tree-depth~$3$ can have at most $7=2^3-1$ vertices}
\end{figure}

If the number of vertices in the longest simple path in a graph $G$ is $\ell$, then $G$ has tree-depth at most $\ell$, for this will be the maximum possible height of a depth-first-search forest. In the other direction, the tree-depth of $G$ must be at least $\log_2(\ell+1)$, because the path itself requires this height in any ancestral forest; see~\autoref{fig:treedepth-b}. Therefore, for any infinite family of graphs, it is equivalent to state that either the family has bounded tree-depth or that it has bounded longest path length~\cite{NesOss-12}.

\subsection{Tree structures for graph connectivity}

Our results in this section depend on a characterization of the planar graphs of bounded tree-depth, which we state in terms of two standard data structures for graph connectivity, the block-cut tree and the SPQR tree.

In any undirected graph $G$, one can define an equivalence relation on the edges of $G$ according to which two distinct edges are equivalent whenever they belong to a simple cycle of $G$. The equivalence classes of this relation form subgraphs of $G$, called the \emph{blocks} or \emph{2-vertex-connected components} of~$G$, and are separated from each other by \emph{cut vertices}, vertices that belong to more than one block. The collection of blocks and cut vertices may be represented by a tree structure (or, in a disconnected graph, a forest) called the \emph{block-cut tree}~\cite{KhuThu-JAlg-93}: this tree has a node for each block and each cut vertex of $G$, and it has an edge connecting each cut vertex to the blocks that contain it. A graph is \emph{2-vertex-connected} if it has only a single block (and no isolated vertices).

Similarly, a graph is $3$-vertex-connected if it does not have any \emph{separation pairs}, pairs of vertices whose removal would disconnect the graph. If a graph~$G$ is 2-connected but not 3-connected, it may be represented by a structure called an SPQR tree~\cite{DiBTam-ICALP-90}. In the formulation of SPQR trees that we use here, there are nodes of three types, each of which is associated with a \emph{triconnected component} of the given graph~$G$, graphs derived from $G$ and having as their vertices subsets of the vertices of~$G$~\cite{HopTar-SJC-73,Mac-DMJ-37}. In an $S$-node of an SPQR tree, the associated graph is a simple cycle. In a $P$-node, the associated graph is a \emph{dipole graph}, a two-vertex multigraph with three or more edges connecting its two vertices. And in an $R$-node of an SPQR tree, the associated graph is a $3$-vertex-connected graph that is neither a cycle nor a dipole. In each of the graphs associated with the nodes of an SPQR tree, some of the edges are \emph{real} (belonging to the original graph~$G$) and some are \emph{virtual}. Each edge of the SPQR tree identifies a pair of virtual edges from the graphs associated with its two nodes, and each virtual edge is identified in this way by one SPQR tree edge. The original graph $G$ can be reconstructed by gluing together each identified pair of virtual edges (merging each endpoint of one edge with an endpoint of the other edge to form a supervertex), and then deleting these edges. With the additional constraints that $S$ nodes not be adjacent to other $S$ nodes, and that $P$ nodes not be adjacent to other $P$ nodes, this decomposition is uniquely determined from~$G$. A 2-connected graph $G$ is planar if and only if each of the graphs associated with the nodes of its SPQR tree is planar, and the SPQR tree can be used to represent all of its possible planar embeddings.

\subsection{Characterization of planar bounded tree-depth}

We use the block-cut tree and SPQR tree structures to characterize the planar graphs of bounded tree-depth, as follows.

\begin{lemma}
A family, $\mathcal{F}$, of planar graphs has bounded tree-depth if and only if the following three conditions are satisfied:
\begin{itemize}
\item The block-cut trees of the graphs in $\mathcal{F}$ have bounded height.
\item The SPQR trees of the blocks of the graphs in $\mathcal{F}$ have bounded height.
\item For each node of an SPQR tree of a block of a graph in $\mathcal{F}$, the number of vertices in the graph associated with the node is bounded.
\end{itemize}
\end{lemma}

\begin{proof}
The requirements that the block-cut trees and SPQR trees have bounded height are clearly both necessary for the tree-depth to be bounded. For, if the graphs in $\mathcal{F}$ had block-cut trees or SPQR trees that could be arbitrarily high, we could follow the structure of these trees to find paths in the graphs that were arbitrarily long, which in turn would imply that the tree-depth (which is at least logarithmic in the path length) would necessarily itself be arbitrarily high. In an SPQR tree of any graph, a P-node cannot have an unbounded number of vertices. An S-node with many vertices would necessarily again lead to a long path, as any virtual edge in the cycle associated with the S-node can be replaced by a path of real edges. To show that it is also necessary to have a bounded number of vertices in the graphs associated with R nodes, we invoke a result of Chen and Yu~\cite{CheYu-JCTB-02}, who proved that a 3-connected planar graph with $k$ nodes necessarily contains a cycle of length $\Omega(k^{\log_2{3}})$. If the graphs associated with R-nodes could have nonconstant values of~$k$, this would again necessarily imply the existence of paths of non-constant length in the graphs of $\mathcal{F}$, and therefore it would also imply that the tree-depth of $\mathcal{F}$ would not be bounded.

In the other direction, suppose that we have a family $\mathcal{F}$ of graphs whose block-cut tree height is at most $b$, SPQR tree height is at most $h$, and SPQR node vertex count is at most $s$. Then, any simple path $\Pi$ in such a graph can pass through at most $b$ blocks: $\Pi$ must follow a path in the block-cut tree, alternating between blocks and cut vertices, and the total number of blocks and cut vertices in any such path is at most $2b-1$.
Within any block of the block-cut tree, $\Pi$ must follow a walk through the SPQR tree that passes through each node at most $s$ times, because each time it passes through a node it uses up one of the vertices in the associated graph. Therefore, the subtree of the SPQR tree followed by $\Pi$ has maximum degree at most $s+1$, and can include at most $\Oh((s+1)^h)$ nodes. Within each of these nodes, $\Pi$ can have at most $s$ vertices per node. Therefore, the total number of nodes in $\Pi$ can be at most $\Oh(sb(s+1)^h)=\Oh(1)$. Any depth-first-search forest of a graph in  $\mathcal{F}$ has height bounded by this quantity, showing that all graphs in $\mathcal{F}$ have bounded tree-depth.
\end{proof}

\subsection{M\"obius gluing}

For any planar graph $G$ of bounded tree-depth, we will form a circle-contact representation of $G$ by decomposing $G$ into its block-cut tree, decomposing each block into its SPQR tree, using Koebe's circle packing theorem to find a circle-contact representation for each node of each SPQR tree, and gluing these packings together to give a representation of $G'$. Because each SPQR tree node has an associated graph of bounded size, the packings that we glue together will be balanced. When we glue two packings together, we will ensure that we only form a polynomial imbalance in the ratio of sizes of their circles. This, together with the bounded height of the trees that guide our overall gluing strategy, will ensure that the eventual packing we construct is balanced.

To glue two packings together, we use \emph{M\"obius transformations}, a group of geometric transformations of the plane (plus one point at infinity) that includes inversion by a circle as well as the familiar Euclidean translations, rotations, and scaling transformations. These transformations take circles to circles (allowing for some degenerate cases in which a line must be interpreted as an infinite-radius circle through the point at infinity) and preserve contacts between circles. The specific transformation we need is described by the following lemma.

\begin{figure}[t]
\centering\includegraphics[width=0.5\textwidth]{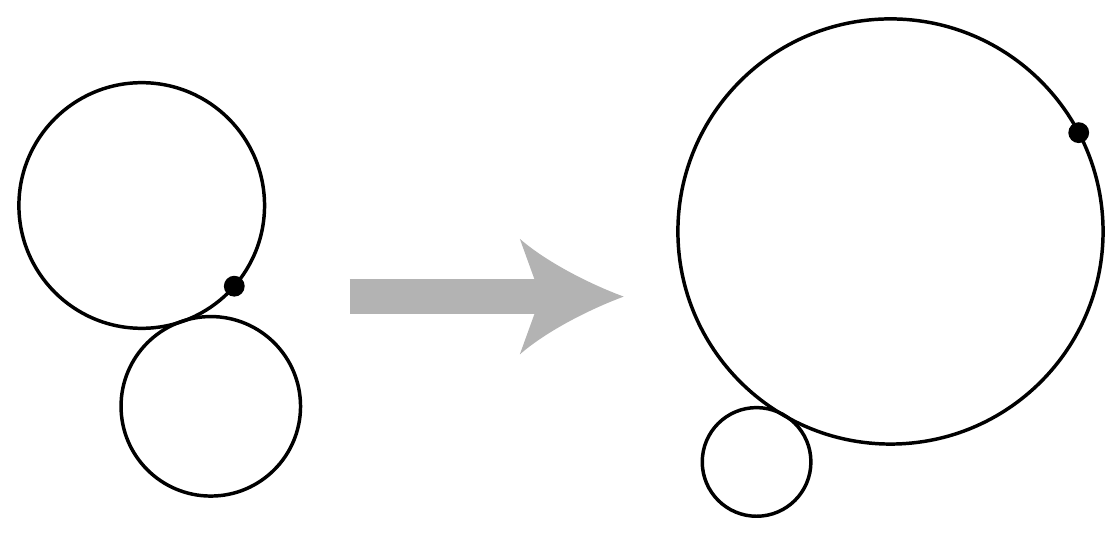}
\caption{\autoref{lem:mobius-matches-marks}: Any two configurations of two tangent circles and a marked point on one of the two circles can be M\"obius-transformed into each other}
\label{fig:tanmark}
\end{figure}

\begin{lemma}
\label{lem:mobius-matches-marks}
Let $X$ and $Y$ be two geometric configurations, each consisting of two tangent circles and a marked point on one but not both of the circles. Then there exists a M\"obius transformation that takes $X$ to $Y$.
\end{lemma}

\begin{proof}
It is equivalent to state that there exists a M\"obius transformation that takes both of these configurations to the same configuration as each other. Both $X$ and $Y$ may be transformed to a pair of parallel lines, by an inversion through a circle centered on their point of tangency. A scaling transformation makes these two pairs of parallel lines the same distance apart from each other, a rotation makes all four lines parallel, a reflection (if necessary) makes both transformed configurations have the marked point on the upper of the two lines, and then a translation matches the positions of the marked points.
\end{proof}

\autoref{lem:mobius-matches-marks} is illustrated in \autoref{fig:tanmark}.
To control the sizes of the transformed packings generated by this lemma, we use the following lemma:

\begin{lemma}
\label{lem:place-marks}
Let two circles $C_0$ and $C_1$ be given, and suppose that we wish to apply \autoref{lem:mobius-matches-marks}
to transform a system of at most $n$ balanced packings, each containing at most $n$ circles, so that some two tangent circles in each packing are placed at the positions of $C_0$ and $C_1$.
Then there exists a sequence of marks on circle $C_0$ such that,
if each packing is placed with one of its points of contact on one of these marks,
then no two of the transformed packings intersect each other,
and the transformed size of the circles in each packing is smaller than the size of $C_0$ and $C_1$ by at most a polynomial factor.
\end{lemma}

\begin{proof}
Place a line tangent to the two circles and place the first mark at the point where it touches $C_0$. Then, to place each successive mark, pack a sequence of $2n$ circles into the triangular gap between $C_0$, $C_1$, and the circle or line defining the previous mark, in such a way that each circle is tangent to $C_0$, $C_1$, and the previous circle in the sequence. Place the next mark at the point where the last of these circles touches $C_0$. These circles are as large as possible for their position in any packing that lies between $C_0$ and $C_1$, so it is not possible for any of the transformed packings to stretch more than half-way from one mark to the next; therefore, no two transformed packings interfere with each other.

To analyze the size of the circles used to generate this system of marks, it is helpful to construct the same sequence of circles in a different way. Place two parallel lines in the plane, a unit distance apart, and pack a sequence of $2n(n-1)+1$ unit-diameter circles between them, so that each of the circles in this sequence is tangent to both lines and to its neighbors in the sequence.
Then, invert this system of lines and circles through a unit-radius circle centered on a boundary point of the first circle in the sequence of unit-diameter circles; see~\autoref{fig:beads}.
This first circle will invert to a line, tangent to the circles $C_0$ and $C_1$ formed by the inverted images of the two parallel lines, and the remaining circles will invert to a sequence of circles between $C_0$ and $C_1$ as constructed above. By choosing the inversion center appropriately, it is possible to choose arbitrarily the ratio of radii of the circles $C_0$ and $C_1$ formed by inverting the two parallel lines, so we can construct a system of circles in this way that is similar to the previous construction for any arbitrary pair of circles $C_0$ and $C_1$.

\begin{figure}[t]
\centering\includegraphics[scale=0.65]{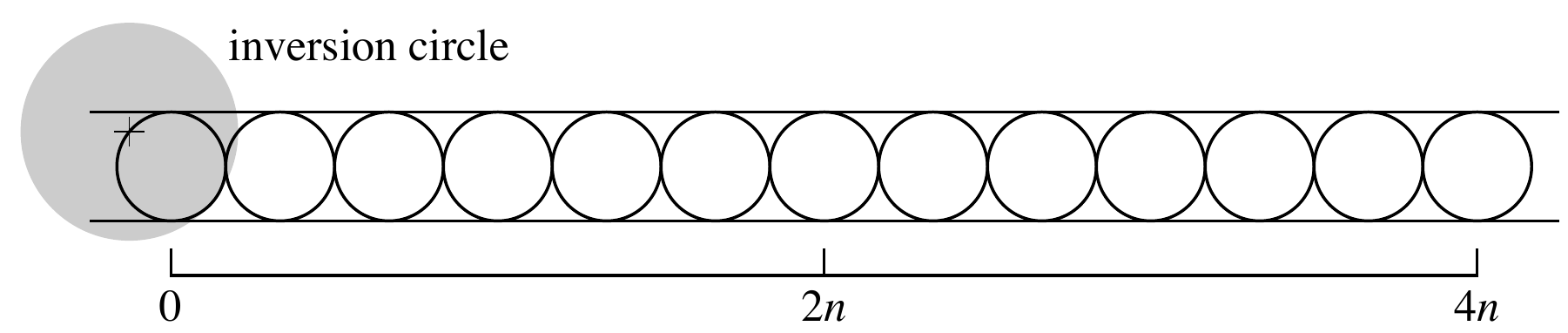}
\caption{An inverted view of the construction in \autoref{lem:place-marks}, for $n=3$}
\label{fig:beads}
\end{figure}

When we invert by a unit circle, the distances of each point from the inversion center become inverted: the closer of the two parallel lines becomes the larger of the two circles $C_0$ and $C_1$, which may be arbitrarily large. Because the other one of the two parallel lines has all points at distance at least $1/2$ from the inversion center, some points within distance $1$ of the inversion center, and some points arbitrarily far from the inversion center, it becomes transformed into a circle that passes from the inversion center to some points outside the inversion circle, but that remains within a radius-$2$ disk centered at the inversion center. Thus, its radius is necessarily between $1/2$ and~$1$. The $i$th of the $2n(n-1)+1$ circles in the sequence of circles has its nearest point to the inversion center at a distance between $i-1/2$ and $i$, and its farthest point to the inversion center at a distance between $i+1/2$ and $i+1$, so when these distances are inverted its transformed diameter becomes $\Th(1/i^2)$. Therefore, the smallest of the $2n(n-1)$ circles used to generate this system of marks is only polynomially small: its radius is $\Th(1/n^4)$ times the smaller of the radii of $C_0$ and $C_1$.
\end{proof}

\subsection{Adjacent separating pairs}

As a warm-up to our main result on bounded tree-depth, we show here the existence of balanced circle-contact representations for a special subclass of the planar graphs of bounded tree-depth, the graphs $G$ in which the separation pairs are all adjacent. That is, if we
construct the block-cut tree of $G$, and for each block construct the SPQR tree,
then each separation pair of each SPQR tree should be the endpoints of an edge somewhere in~$G$. Equivalently, among each two adjacent nodes in any of the SPQR trees for the blocks of~$G$, one of the two nodes should be a P-node, and one of the edges of this P-node should be a real edge.

\begin{lemma}
\label{lem:adj-separation}
For every constant bound $d$, every planar graph~$G$ with tree-depth at most $d$ in which all separation pairs are adjacent has a balanced circle-contact representation.
\end{lemma}

\begin{proof}
We form the block-cut tree of $G$, form the SPQR tree of each block, and form a circle-contact representation for the (constant size) graph associated with each SPQR tree node. For the R nodes, we use the circle packing theorem; for the S nodes, we pack equal-size circles in a circular layout; and for the P nodes, we place two equal circles in contact with each other (representing the one real edge of the P node).
We must then glue these packings together so that, when two or more nodes contain the same separation pair, this pair is represented by the same two circles; because of our assumption about the adjacency of separation pairs, we can assume that these circles are tangent to each other.

To perform this gluing step, at one of the P-nodes of the SPQR tree, we consider all of the neighboring nodes connected to it (corresponding to the virtual edges of the P node). For each neighboring node $X$, we form a configuration consisting of the two circles representing the separation vertices in its packing, and a marked point where another circle is tangent to these two circles.  For the P-node itself, we have the corresponding two circles, on which we may place a mark at an arbitrary point. By \autoref{lem:mobius-matches-marks}, we may transform the packing representing $X$ (and anything else that has already been glued to it) so that its separation pair is represented by the same two circles as the $P$ node and its marked point of contact with these two circles is at the given point. By placing the marked point on the P-node sufficiently close to its two circles' point of tangency, the packing representing $X$ will also be transformed to be close to the point of tangency, allowing it to avoid overlapping with any other part of the circle packing that has already been glued to the P-node.

In this way, the packings for each node of an SPQR tree can be glued together, step by step, to form a single packing representing the whole SPQR tree. At each P-node, in order to keep the packings glued to that node from interfering with each other, it is sufficient to place the marks polynomially close to each other and to the point of tangency, causing the glued-in packings to be reduced in size (compared to the two circles representing the separation vertices) by a polynomial factor. Because the SPQR tree has bounded height, at most a constant number of these polynomial factors are compounded together, resulting in polynomial balance overall.

The process for gluing together blocks of the block-cut tree is very similar, but simpler, because each gluing step must match only one circle (the articulation vertex between two blocks) rather than two circles. The analysis is the same.
\end{proof}

\subsection{Inversive-distance circle packings}

To apply the same method to graphs with non-adjacent separation pairs, we would like to modify the graphs associated with each SPQR tree node by removing the virtual edges between each non-adjacent separation pair, before finding a circle-contact representation for each of these graphs and then gluing these representations together by M\"obius transformations.
In the case of adjacent separation pairs, this gluing used \autoref{lem:mobius-matches-marks}, in which a one-parameter family of M\"obius transformations allows any tangent pair of circles to be mapped to any other tangent pair of circles, with an additional degree of freedom that allows two marked points to be aligned with each other. However, for non-adjacent separation pairs, the problem is made more complicated by the fact that pairs of disjoint circles cannot always be mapped into each other by a M\"obius transformation. Instead, an invariant of these pairs called the \emph{inversive distance} controls the existence of such a mapping. Two Euclidean circles with radii $r$ and $R$, whose centers are at distance $\ell$ from each other, have inversive distance
$$I=\frac{\ell^2-r^2-R^2}{2rR}.$$
This is an invariant under M\"obius transformations, and can also be defined for circles in hyperbolic and spherical geometry~\cite{BowHur-VM3-03}. The inversive distance of two tangent circles is one; it is less than one for circles that cross each other and greater than one for disjoint circles.
Two pairs of circles that both have equal inversive distances can be mapped into each other in the same way as \autoref{lem:mobius-matches-marks} (with the freedom to align two marked points) but two pairs with unequal inversive distances cannot be mapped into each other.

Bowers and Stephenson~\cite{BowSte-UD-04} suggested the use of inversive distance in circle packings. An inversive-distance circle packing is specified by a planar graph $G$ and a label $I(uv)$ for each edge $uv$ of $G$; the goal is to represent the vertices of $G$ by circles such that, for each edge $uv$, the circles representing $u$ and $v$ are at inversive distance exactly $I(uv)$ from each other. Bowers and Stephenson write that ``the theoretical underpinnings are not yet in place'' for this type of packing. In particular it is still not known under what conditions an inversive-distance circle packing exists, although a numerical method briefly suggested by Collins and Stephenson~\cite{ColSte-CGTA-03} and detailed by Bowers and Hurdal~\cite{BowHur-VM3-03} is reported to work well in practice for finding such packings. Notwithstanding this gap in our theoretical knowledge, we will prove that, in the case of interest for us, inversive-distance circle packings do exist and can be used to find balanced packings of graphs of bounded tree-depth.

One thing that is known about inversive-distance circle packings is that, when they do exist for a given maximal planar graph, they are unique.

\begin{lemma}[Luo~\cite{Luo-GT-11}]
\label{lem:idcp-uniqueness}
Let $G$ be a maximal planar graph, with outer face $\Delta$, let $G$ have an assignment $I$ of inversive distances to its edges, and fix a non-collinear placement of the three vertices of $\Delta$ in the Euclidean plane. If an inversive-distance circle packing exists for $G$ and $I$, with the circles representing $\Delta$ centered at the fixed placement of the vertices of $\Delta$, then this packing is uniquely determined by $G$, $I$, and the placement of~$\Delta$.
\end{lemma}

The fixed placement of $\Delta$ simplifies the statement of this result by eliminating the possibility of M\"obius transformations that would change the position of the packing without changing its inversive distances.
Luo actually proved this result in greater generality, for triangulated topological surfaces with Euclidean or hyperbolic geometry; the analogous statement for spherical geometry turns out not to be true~\cite{MaSch-DCG-12}. In Luo's more general setting, the packing is determined (up to scale) by the system of angular defects at each vertex of the surface. In the version that we need for the Euclidean plane, the angular defects are zero except at the vertices of $\Delta$, and the fixed placement of $\Delta$ determines both the angular defects at these vertices and the scale of the packing, so the result as stated above follows from Luo's more general theorem.

\begin{corollary}
\label{cor:idcp-open}
Let $G$ be a maximal planar graph, with outer face $\Delta$, and fix a non-collinear placement of the three vertices of $\Delta$ in the Euclidean plane. Let $m$ be the number of edges in~$G$. Then the set of assignments of distances to the edges of~$G$ that have inversive-distance circle packings is an open subset of $\RR^m$, and there exists a continuous function from these distance assignments to the circle centers and radii of the (unique) corresponding packings.
\end{corollary}

\begin{proof}
The space of distance assignments is $m$-dimensional, as is (by Euler's formula) the space of possible placements of circle centers (obeying the fixed placement of $\Delta$) and radii in a packing. Thus, by \autoref{lem:idcp-uniqueness}, the function from circle placements to inversive distances is a continuous injective function between two sets of the same dimension. By Brouwer's theorem on invariance of domain~\cite{Bro-MA-12}, this function has an open image and a continuous inverse function.
\end{proof}

By using the fact that the set of allowable inversive-distance assignments is open, we can find the packing we need:

\begin{lemma}
\label{lem:idcp-existence}
Let $G=(V,E)$ be a planar graph, and let $F\subset E$ be a subset of ``virtual'' edges in $G$.
Then, for all sufficiently small $\eps>0$, there exists a collection of circles in the plane, and a one-to-one correspondence between these circles and the vertices of $G$, such that:
\begin{itemize}
\item For each edge $e\in E\setminus F$, the two circles representing the endpoints of $e$ are tangent (inversive distance~$1$), and
\item for each edge $e\in F$, the two circles representing the endpoints of $e$ have inversive distance exactly $1+\eps$.
\end{itemize}
Additionally, this circle packing can be chosen to vary continuously with $\eps$, in such a way that as $\eps\to 0$ the packing converges to a circle-contact representation for~$G$.
\end{lemma}

\begin{proof}
Add extra vertices to $G$, and extra edges incident to those vertices, to make the augmented graph maximal planar; include as part of this augmentation a triangle $\Delta$ of added vertices.
Choose a placement of $\Delta$ at the vertices of an equilateral triangle in the plane.
Let $D$ be the distance assignment that sets all inversive distances on edges of the augmented graph to equal~$1$; then there exists a circle packing for~$D$, by the original circle packing theorem of Koebe.
By \autoref{cor:idcp-open}, there exists an open neighborhood of $D$ within which all distance assignments are achievable; let $\eps$ be sufficiently small that all vectors within $L_\infty$ distance $\eps$ of $D$ are achievable.
In particular, the distance assignment that sets all distances in $F$ equal to $1+\eps$, and all remaining distances equal to~$1$, has an inversive circle packing. The subset of this packing that corresponds to the original vertices in $V$ then satisfies the conditions of the lemma. Once the augmentation of $G$ and the placement of $\Delta$ is fixed, the continuous variation of the packing on $\eps$ follows from the continuity of the function from distance assignments to packings stated in \autoref{cor:idcp-open}.
\end{proof}

\subsection{Circle packings for graphs of bounded tree-depth}

\begin{theorem}
For every constant bound $d$, every planar graph with tree-depth at most $d$ has a balanced circle-contact representation.
\end{theorem}

\begin{proof}
We follow the same outline as the proof of \autoref{lem:adj-separation}, except that we choose a sufficiently small value of $\eps$ (the same $\eps$ for all nodes of all SPQR trees)
and use \autoref{lem:idcp-existence} to construct circle packings for each node in which the real edges of the graph associated with the node, and the virtual edges corresponding to adjacent separation pairs, have inversive distance one (that is, their circles are tangent) while the remaining edges in each node (non-adjacent separation pairs) have inversive distance exactly $1+\eps$ (their circles are disjoint).

The analogue of \autoref{lem:mobius-matches-marks} holds for any two pairs of tangent circles with the same inversive distance as each other. By continuity of the packing with respect to $\eps$, for small enough values of $\eps$, the analogues of \autoref{lm:bounded-packing} and \autoref{lem:place-marks} also remain valid. Additionally, although inversive-distance circle packings in general are not guaranteed to have disjoint circles for pairs of vertices that are not connected by an edge in the underlying graph, this property also follows for our circle packings, for sufficiently small $\eps$, by continuity. With these ingredients in hand, the proof proceeds as before.
\end{proof}

We remark that the proof described above does not require $\eps$ to be significantly smaller than the balance of the resulting packing. Indeed, we only need to apply the inversive circle packing method of \autoref{lem:idcp-existence} to the R-nodes of SPQR trees, as the other nodes have associated graphs that are easy to represent directly. Because each R-node has $\Oh(1)$ vertices, there are only $\Oh(1)$ combinatorially distinct choices of an associated graph and set of virtual edges for which we need packings, and in order to apply \autoref{lem:idcp-existence} we can choose any $\eps$ that is small enough to work for each of these inputs, a value that is independent of~$n$. The limiting factor controlling our actual choice of $\eps$ is  the variant of \autoref{lem:place-marks} that applies to inversive-distance circle packings, which only requires $\eps$ to be polynomially small.

Although our proof is not constructive, the report by Collins and Stephenson~\cite{ColSte-CGTA-03} that their inversive-distance circle packing algorithm works well in practice gives us hope that the same would be true of its application to this problem.

\fi

\section{Conclusion}

We studied balanced circle packings for planar graphs, showing that several
rich classes of graphs have balanced circle
\ifFull
packings, while several simple
classes of graphs do not.
\fi
One interesting open problem is whether or not
every outerplanar graph has a balanced circle packing representation.
While we identified several subclasses of outerplanar graphs that admit
such representations, the question remains open for general outerplanar graphs.

\medskip\noindent {\bf Acknowledgments.}
This work is supported
in part by the National Science Foundation under
grants CCF-1228639, CCF-1115971, DEB 1053573, and by the Office of Naval Research
under Grant No. N00014-08-1-1015.

\ifFull
{\raggedright
\bibliographystyle{abuser}}
\fi
\bibliography{circle-contact}

\end{document}